\documentclass[11pt]{article}
\usepackage{amssymb,amsmath,amsthm,sectsty,amsfonts}
\usepackage{float} 

\usepackage{fullpage} 
\usepackage[pdfstartview=FitH,colorlinks,linkcolor=blue,filecolor=blue,citecolor=blue,urlcolor=blue,pagebackref=false]{hyperref}
\usepackage{upgreek} 

\usepackage{xspace} 
\usepackage{cite} 

\usepackage{hhline}
\usepackage{color} 
\usepackage{nicefrac} 
\usepackage{paralist} 

\usepackage[english]{babel}

\usepackage{xifthen}
\usepackage{color}
\usepackage{graphicx}
\usepackage[countmax]{subfloat}
\usepackage[noend]{algorithmic}

\usepackage[all,line]{xy}

\newcommand{\es}{\varnothing} 
\newcommand{\sm}{\setminus}
\newcommand{\se}{\subseteq}
\newcommand{\Q}{\cQ}
\newcommand{\bQ}{\bfQ}

\newcommand{\bB}{\bfB}

\newcommand{\bM}{\bfM}
\newcommand{\bP}{\bfP}

\newcommand{\RED}{\mathsf{R}}
\newcommand{\GRA}{\mathsf{GA}}
\newcommand{\GRB}{\mathsf{GB}}
\newcommand{\GR}{\mathsf{G}}

\newcommand{\win}{\mathtt{win}}

\newcommand{\Dep}{\mathsf{StatDep}}
\newcommand{\SelDep}{\mathsf{SelfDep}}

\newcommand{\remove}[1]{}
\newcommand{\num}[1]{{\bf (#1)}}

\newcommand{\ol}{\overline}
\newcommand{\wt}[1]{\widetilde{#1}}

\newcommand{\nin}{\not \in}

\newcommand{\Good}{\mathsf{Good}} 
\newcommand{\BGood}{\mathsf{BGood}}
\newcommand{\AGood}{\mathsf{AGood}}
\newcommand{\Fail}{\mathsf{Fail}}
\newcommand{\BFail}{\mathsf{BFail}}
\newcommand{\AFail}{\mathsf{AFail}}

\newcommand{\LAST}{\mathsf{LAST}}

\newcommand{\EXEC}{\mathcal{E}} 
\newcommand{\GViews}{\mathcal{GV}}
\newcommand{\Views}{\mathcal{V}}

\newcommand{\etal}{et~al.\ }
\newcommand{\wolog}{w.l.o.g\ }
\newcommand{\Wolog}{W.l.o.g\ }

\newcommand{\resp}{resp.\ }
\newcommand{\ie}{i.e.,\ }
\newcommand{\eg}{e.g.,\ }
\newcommand{\wrt} {with respect to\ }

\newcommand{\set}[1]{\left\{ #1 \right\}}

\newcommand{\bits}{\{0,1\}}

\newcommand{\size}[1]{\left|#1\right|}
\newcommand{\abs}[1]{\size{#1}}

\newcommand{\To}{\mapsto}

\newcommand{\R}{{\mathbb R}}

\newcommand{\cC}{{\mathcal C}}
\newcommand{\cD}{{\mathcal D}}

\newcommand{\cP}{{\mathcal P}}
\newcommand{\cQ}{{\mathcal Q}}

\newcommand{\cU}{{\mathcal U}}

\newcommand{\bfA}{\mathbf{A}}
\newcommand{\bfB}{\mathbf{B}}

\newcommand{\bfH}{\mathbf{H}}

\newcommand{\bfM}{\mathbf{M}}

\newcommand{\bfP}{\mathbf{P}}
\newcommand{\bfQ}{\mathbf{Q}}

\newcommand{\bfU}{\mathbf{U}}
\newcommand{\bfV}{\mathbf{V}}
\newcommand{\bfW}{\mathbf{W}}

\newcommand{\bfa}{\mathbf{a}}
\newcommand{\bfb}{\mathbf{b}}

\newcommand{\bfg}{\mathbf{g}}

\newcommand{\bfr}{\mathbf{r}}

\newcommand{\bfx}{\mathbf{x}}
\newcommand{\bfy}{\mathbf{y}}
\newcommand{\bfz}{\mathbf{z}}

\newcommand{\eps}{\varepsilon}
\newcommand{\e}{\varepsilon}

\newcommand{\poly}{\operatorname{poly}}

\newcommand{\Exp}{\operatorname*{\mathbb{E}}}
\newcommand{\Ex}{\Exp}

\newcommand{\negl}{\operatorname{neg}}
\newcommand{\Supp}{\operatorname{Supp}}

\newcommand{\class}[1]{\ensuremath{\mathbf{#1}}}

\newcommand{\PSPACE}{\class{PSPACE}}

\newtheorem{theorem}{Theorem}[section]

\theoremstyle{plain}
\newtheorem{claim}[theorem]{Claim}
\newtheorem{proposition}[theorem]{Proposition}
\newtheorem{lemma}[theorem]{Lemma}
\newtheorem{corollary}[theorem]{Corollary}

\theoremstyle{definition}
\newtheorem{definition}[theorem]{Definition}
\newtheorem{construction}[theorem]{Construction}
\newtheorem{protocol}[theorem]{Protocol}

\theoremstyle{definition}
\newtheorem{remark}[theorem]{Remark}

\newcommand{\namedref}[2]{#1~\ref{#2}}

\newcommand{\torestate}[3]{%
\expandafter \def \csname BBRESTATE #2 \endcsname{#3}
\theoremstyle{plain}
\newtheorem{BBRESTATETHMNUM#2}[theorem]{#1}
\begin{BBRESTATETHMNUM#2}\label{#2}\csname BBRESTATE #2 \endcsname   \end{BBRESTATETHMNUM#2}
\newtheorem*{BBRESTATETHMNONNUM#2}{\namedref{#1}{#2}}
}

\newcommand{\restate}[1]{\begin{BBRESTATETHMNONNUM#1}[Restated] \csname BBRESTATE #1 \endcsname
\end{BBRESTATETHMNONNUM#1}}

\title{Merkle's Key Agreement Protocol is Optimal: \\
An $O(n^2)$ Attack on any Key Agreement from Random Oracles}

\author{Boaz Barak\thanks{Microsoft Research New England and Harvard University,
    \texttt{b@boazbarak.org}.} \and
Mohammad Mahmoody\thanks{University of Virginia, \texttt{mohammad@cs.virginia.edu}. Supported by NSF CAREER award CCF-1350939.}
}
\date{}
\begin{document}

\maketitle

\begin{abstract}
We prove that every key agreement protocol in the random oracle model in which the honest users make at most $n$ queries
to the oracle can be broken by an adversary who makes $O(n^2)$ queries to the oracle. This improves on the previous
$\Tilde{\Omega}(n^6)$ query attack given by Impagliazzo and Rudich (STOC '89) and resolves an open question posed by
them.

Our bound is optimal up to a constant factor since Merkle  proposed a key agreement protocol in 1974 that can be easily
implemented with $n$ queries to a random oracle and cannot be broken by any adversary who asks $o(n^2)$ queries.

\end{abstract}

\thispagestyle{empty}
%

%
\newpage
\tableofcontents

\thispagestyle{empty}

\clearpage
\setcounter{page}{1}

\section{Introduction}

In the 1970's Diffie, Hellman, and Merkle~\cite{Merkle74,DiffieHe76,Merkle78} began to challenge the accepted wisdom that two parties cannot communicate
confidentially over an open channel without first exchanging a secret key using some secure means. The first such
protocol (at least in the open scientific community) was proposed by Merkle in 1974~\cite{Merkle74} for a course project in Berkeley. Even though the course's instructor rejected the proposal, Merkle continued working on his ideas and discussing them with Diffie and Hellman, leading to the papers~\cite{DiffieHe76,Merkle78}. Merkle's original key exchange protocol was extremely simple and can be directly formalized and implemented using a random oracle\footnote{In this work, \emph{random oracles} denote any randomized oracle $O \colon \bits^* \To \bits^*$ such that $O(x)$ is independent of $O(\bits^* \sm \set{x})$ for every $x$ (see Definition~\ref{def:RandOr}). The two protocols of Merkle we describe here can be implemented using a length-preserving random oracle (by cutting the inputs and the output to the right length). Our \emph{negative} results, on the other hand, apply to any random oracle.
} as follows:

\begin{protocol}[Merkle's 1974 Protocol using Random Oracles] \label{prot:merkle}
Let $n$ be the security parameter and $H: [n^2] \To \bits^{n}$ be a function chosen at random accessible  to all parties as an oracle. Alice and Bob execute the protocol as follows.
\begin{enumerate}

\item Alice chooses $10n$ distinct random numbers $x_1,\ldots,x_{10n}$ from $[n^2]$ and sends $a_1,\ldots,a_{10n}$ to Bob
    where $a_i=H(x_i)$.

\item Similarly, Bob chooses $10n$ random numbers $y_1,\ldots,y_{10n}$ in $[n^2]$ and sends $b_1,\ldots,b_{10n}$ to Alice
    where $b_j=H(y_j)$. (This step can be executed in parallel with Alice's first step.)

\item If there exists any $a_i=b_i$ among the exchanged strings, Alice and Bob let $(i,j)$ to be the lexicographically first index of such pair; Alice takes $x_i$ as her key and Bob takes $y_j$ as his key. If no such $(i,j)$ pair exits, they both take $0$ as the agreed key.

\end{enumerate}

\end{protocol}

It is easy to see that with probability at least $ 1-n^4/2^n$, the random function $H: [n^2] \To \bits^{n}$ is injective, and so any $a_i=b_i$ will lead to the same key $x_i=y_j$ used by Alice and Bob. In addition, the probability of \emph{not} finding a ``collision'' $a_i=b_j$ is at most $(1-10/n)^{10n} \leq (1/e)^{100} < 2^{-100}$ for all $n \geq 10$. Moreover, when there is a collision $a_i=b_j$, Eve has to essentially search the whole input space $[n^2]$ to find the preimage $x_i=y_j$ of $a_i=b_j$ (or, more precisely, make $n^2/2$ calls to $H(\cdot)$ on average).

We note that in his 1978 paper~\cite{Merkle78} Merkle described a different variant of a key agreement protocol by having Alice send to Bob $n$ ``puzzles'' $a_1,\ldots, a_n$ such that each puzzle $a_i$ takes $\approx n$ ``time'' to solve (where the times is modeled as the number of oracle queries), and the solver learns some secret $x_i$. The idea is that Bob would choose at random which puzzle $i\in [n]$ to solve, and so spend $\approx n$ time to learn $x_i$ which he can then use as a shared secret with Alice after sending a hash of $x_i$ to Alice so that she knows which secret Bob chose. On the other hand, Eve would need to solve almost all the puzzles to find the secret, thus spending $\approx n^2$ time.
These puzzles can indeed be implemented via a random oracle $H\colon [n]  \times [n] \To \bits^{n} \times \bits^m$ as follows. The $i$'th puzzle with hidden secret $x\in\bits^m$ can be obtained by choosing  and $k\gets [n]$ at random and getting $a_i = (H_1(i,k), H_2(i,k) \oplus x)$ where $\oplus$ denotes bitwise exclusive OR, $H_1(\cdot,\cdot)$ denotes the first $n$ bits of $H$'s output, and $H_2(\cdot,\cdot)$ denotes the last $m$ bits of $H$'s output. Now, given puzzles $P_1=(h^1_1,h^2_2),\dots,P_n=(h^n_1,h^n_2)$, Bob takes a random puzzle $P_j$, solves it by asking $H(j,k)$ for all $k \in [n]$ to get $H(j,k)=(h^j_1,h_2)$ for some $h_2$, and then he retrieves the puzzle solution $x=h_2 \oplus h^j_2$.


%

One problem with Merkle's protocol is that its security was only analyzed in the random oracle model which does not
necessarily capture security when instantiated with a cryptographic one-way or hash function~\cite{CanettiGoHa98}.
Biham, Goren, and Ishai~\cite{BihamGoIs08} took a step towards resolving this issue by providing a security
analysis for Merkle's protocol under  concrete complexity assumptions. In particular, they proved that assuming the existence of one-way functions that cannot be inverted with
probability more than $2^{-\alpha n}$ by adversaries running in time $2^{\alpha n}$ for $\alpha \geq 1/2-\delta$, there
is a key agreement protocol in which Alice and Bob run in time $n$ but any adversary whose running time is at most
$n^{2-10\delta}$ has $o(1)$ chance of finding the secret.

Perhaps a more serious issue with
Merkle's protocol is that it only provides a \emph{quadratic} gap between the running time of the honest parties and
the adversary. Fortunately, not too long after Merkle's work, Diffie and Hellman~\cite{DiffieHe76} and later Rivest,
Shamir, and Adleman~\cite{RivestShAd78} gave constructions for key agreement protocols that are conjectured to have
\emph{super-polynomial} (even subexponential) security and are of course widely used to this day. But because these and later protocols are based on certain
algebraic computational problems, they could perhaps be vulnerable to unforseen attacks using this algebraic structure.
It remained, however, an important open question to show whether there exist key agreement protocols with super-polynomial
security that use only a random oracle.\footnote{This is not to be confused with some more recent works such as
\cite{BellareRo93}, that combine the random oracle model with assumptions on the intractability of other problems such
as factoring or the RSA problem to obtain more efficient cryptographic constructions.} The seminal paper of Impagliazzo
and Rudich~\cite{ImpagliazzoRu89} answered this question negatively by showing that every key agreement protocol,  even in its full general form that is  allowed to run in polynomially many rounds,  can be broken by an adversary asking $O(n^6\log n)$ queries if the two parties ask
$n$ queries in the random oracle model.\footnote{More
accurately,~\cite{ImpagliazzoRu89} gave an $O(m^6\log m)$-query attack where $m$ is the maximum of the number of
queries $n$ and the number of communication rounds, though we believe their analysis could be improved to an $O(n^6\log
n)$-query attack. For the sake of simplicity, when discussing~\cite{ImpagliazzoRu89}'s results we will assume that $m =
n$, though for our result we do not need this assumption.} A random oracle is in particular a one-way function (with
high probability)\footnote{The proof of this statement for the case of non-uniform adversaries is quite nontrivial; see \cite{GennaroGeKaTr05} for a proof.}, and thus an important corollary of~\cite{ImpagliazzoRu89}'s result is that there is no construction of  key agreement protocols
based on one-way functions with a proof of super-polynomial security that is of the standard black-box type
(i.e., the implementation of the protocol uses the one-way function as an oracle, and its proof of security uses the one-way function and  any
adversary breaking the protocol also as oracles).\footnote{This argument applies to our result as well, and of course extends to any other primitive that is implied by random oracles (\eg collision-resistant hash functions) in a black-box way.}

\paragraph{Question and Motivation.} Impagliazzo and Rudich~\cite[Section~8]{ImpagliazzoRu89} mention as an
open question (which they attribute to Merkle) to find out whether their attack can be improved to $O(n^2)$ queries
(hence showing the optimality of Merkle's protocol in the random oracle model) or there exist key agreement protocols in
the random oracle model with $\omega(n^2)$ security. Beyond just being a natural question, it also has some practical
and theoretical motivations. The practical motivation is that protocols with sufficiently large polynomial gap could be
secure enough in practice---e.g., a key agreement protocol taking $10^9$ operations to run and $(10^9)^6 = 10^{54}$
operations to break could be good enough for many applications.\footnote{These numbers are just an example, and in practical applications the constant terms will make an important difference; however we  note that these particular constants are not ruled out by~\cite{ImpagliazzoRu89}'s attack but are ruled out by ours by taking number of
operations to mean the number of calls to the oracle.} In fact, as was argued by Merkle himself~\cite{Merkle74}, as technology
improves and honest users can afford to run more operations, such polynomial gaps only become more useful since the ratio between the work required
by the attacker and the honest user will grow as well. Thus, if
known algebraic key agreement protocols were broken, one might look to polynomial-security protocol such as Merkle's for
an alternative. Another motivation is theoretical--- Merkle's protocol has very limited interaction (consisting of one
round in which both parties simultaneously broadcast a message) and in particular it implies a public key encryption
scheme.  It is natural to ask whether more interaction can help achieve some polynomial advantage over this simple
protocol. Brakerski \etal~\cite{BrakerskiKSY11} show a
simple $O(n^2)$-query attack for protocols with \emph{perfect completeness} based on a random oracles,\footnote{We are not aware of any perfectly complete $n$-query key agreement protocol in the random oracle with $\omega(n)$ security. In other words, it seems conceivable that all such protocols could be broken with a \emph{linear} number of queries.} where the probability is over both the oracle and parties' random seeds. In this work we focus on the main question of~\cite{ImpagliazzoRu89} in full fledged form.

\subsection{Our Results}

In this work we answer the above question of~\cite{ImpagliazzoRu89}, by showing that
every protocol in the random oracle model where Alice and Bob make $n$ oracle queries can be broken with high
probability by an adversary making $O(n^2)$ queries. That is, we prove the following:

\begin{theorem}[Main Theorem] \label{thm:main}
Let $\Pi$ be a two-party protocol in the random oracle model such that when executing
$\Pi$ the two parties Alice and Bob make  at most $n$ queries each, and their outputs are identical with
probability at least $\rho$. Then, for every $0 < \delta < \rho$, there is an eavesdropping adversary Eve making $O(n^2/\delta^2)$ queries to the oracle
whose output agrees with Bob's output with probability at least $\rho-\delta$.
\end{theorem}

To the best of our
knowledge, no better bound than the $\wt{O}(n^6)$-query attack of~\cite{ImpagliazzoRu89} was previously known even
in the case where one does not assume the one-way function is a random oracle (which would have made the task of proving a
negative result easier).

In the original publication of this work~\cite{BarakM09}, the following technical result (Theorem \ref{thm:indep}) was implicit in the proof of Theorem~\ref{thm:main}. Since this particular result has found uses in subsequent works to the original publication of this work~\cite{BarakM09}, here we state and prove it explicitly. This theorem, roughly speaking, asserts that by running the attacker of Theorem~\ref{thm:main} the ``correlation'' between the ``views'' of Alice and Bob (conditioned on Eve's knowledge) remains close to zero at all times. The view of a party consists of the  information they posses at any moment during the execution of the protocol: their private randomness, the public messages, and their private interaction with the oracle.

\begin{theorem}[Making Views almost Independent---Informal] \label{thm:indep}
Let $\Pi$ be a two-party protocol in the random oracle model such that when executing
$\Pi$ the two parties Alice and Bob make  at most $n$ oracle queries each. Then for any $\alpha,\beta<1/10$ there is an eavesdropper Eve making $\poly(n/(\alpha\beta))$ queries to the oracle such that with probability at least $1-\alpha$ the following holds at the end of \emph{every} round: the joint distribution of Alice's and Bob's views so far conditioned on Eve's view is $\beta$-close to being independent of each other.

\end{theorem}

See Section~\ref{sec:Extensions} for the formal statement and proof of  Theorem~\ref{thm:indep}.

\subsection{Related Work} \label{subsec:related}

\paragraph{Quantum-Resilient Key Agreement.}  In one central scenario in which some algebraic key agreement protocols will be broken---the construction of practical quantum computers--- Merkle's protocol will also be broken with linear oracle queries using Grover's search algorithm~\cite{Grover96}. In the original publication of this work we asked whether our $O(n^2)$-query classical attack could lead to an $O(n)$ quantum attack against any classical protocol (where Eve accesses the random oracle in a superposition).  We note that using quantum communication there is an \emph{information theoretically} secure key agreement
protocol~\cite{BennetBrEk92}. Brassard and Salvail~\cite{BrassardSa08} (independently observed by~\cite{BihamGoIs08}) gave a quantum
version of Merkle's protocol, showing that Alice and Bob can use quantum computation (but classical communication)
to obtain a key agreement protocol with super-linear $n^{3/2}$ security in the random oracle model against
quantum adversaries. Finally, Brassard \etal~\cite{BrassardHKKLS11} resolved our question negatively by presenting a \emph{classical} protocol in the random oracle model with super linear security $\Omega(n^{3/2-\eps})$ for arbitrary small constant $\eps$.


\paragraph{Attacks in Small Parallel Time.} Mahmoody, Moran, and Vadhan~\cite{MahmoodyMV11} showed how to improve the \emph{round complexity} of the attacker of Theorem~\ref{thm:main} to  $n$ (which is optimal) for the case of one-message protocols, where a round here refers to a set of queries that are asked to the oracle in parallel.\footnote{For example, a \emph{non-adaptive} attacker who prepares all of its oracle queries and then asks them in one shot, has round complexity one.} Their result rules out constructions of ``time-lock puzzles'' in the \emph{parallel} random oracle model in which the polynomial-query solver needs more \emph{parallel time} (i.e., rounds of parallel queries to the random oracle) than the puzzle generator to solve the puzzle. As an application back to our setting,~\cite{MahmoodyMV11} used the above result and showed that every $n$-query (even multi-round) key agreement protocol can  be broken by $O(n^3)$ queries in only $n$ rounds of oracle queries, improving the $\Omega(n^2)$-round attack of our work by a factor of $n$. Whether an $O(n)$-round $O(n^2)$-query attack is possible remains as an intriguing  open question.

\paragraph{Black-Box Separations and the Power of Random Oracle.} The work of Impagliazzo and Rudich~\cite{ImpagliazzoRu89} laid down the framework for the field of \emph{black-box separations}. A black-box separation of a primitive $\cQ$ from another primitive $\cP$ rules out any construction of $\cQ$ from $\cP$ as long as it treats the primitive $\cP$ and the adversary (in the security proof) as oracles. We refer the reader to the excellent survey by Reingold \etal~\cite{ReingoldTrVa04} for the formal definition and its variants. Due to the abundance of black-box techniques in cryptography, a black-box separation indicates a major disparity between how hard it is to achieve $\cP$ vs. $\cQ$, at least with respect to black-box techniques. The work of \cite{ImpagliazzoRu89} employed the so called ``oracle separation'' method to derive their black-box separation. In particular, they showed that relative to the oracle $O=(R,\PSPACE)$ in which $R$ is a random oracle  one-way functions exist  (with high probability)  but secure key agreement does not. This existence of such an oracle implies a black-box separation.

The main technical step in the proof of \cite{ImpagliazzoRu89} is to show that relative to a random oracle $R$, any key agreement protocol could be broken by an adversary who is computationally unbounded and asks at most $S=\poly(n)$ number of queries (where $n$ is the security parameter). The smallest such polynomial $S$ for any construction $\cC$ could be considered as a quantitative black-box security for $\cC$ in the random oracle model. This is indeed the setting of our paper, and we study the optimal black-box security of key agreement in the random oracle model.
Our Theorem~\ref{thm:main} proves that $\Theta(n^2)$ is the optimal security one can achieve for an $n$-query key agreement protocol in the random oracle model. The techniques used in the proof of Theorem~\ref{thm:main} have found  applications in  the contexts of black-box separations and black-box security in the random oracle model (see, \eg~\cite{KatzSY11,BrakerskiKSY11,MahmoodyP12}).  In the following we describe some of the works that focus on the power of random oracles in secure two-party computation.

Dachman-Soled \etal~\cite{DachmanLMM11} were the first to point out that results implicit in our proof of Theorem~\ref{thm:main} in the original publication of this work~\cite{BarakM09} could be used to show the existence of eavesdropping attacks that gather enough information from the oracle in a way that conditioned on this information  the views of Alice and Bob become ``close'' to being independent (see Lemma 5 of~\cite{DachmanLMM11}). Such results were used in~\cite{DachmanLMM11},~\cite{MahmoodyMP12}, and~\cite{HaitnerHoReSe07} to explore the power of random oracles in secure two-party computation.
Dachman-Soled \etal showed that ``optimally-fair'' coin tossing protocols~\cite{Cleve1986} cannot be based on one-way functions with $n$ input and $n$ output bits in a black-box way if the  protocol has $o(n/\log n)$ rounds.

Mahmoody, Maji, and Prabhakaran~\cite{MahmoodyMP12} proved that random oracles are useful for secure two-party computation of finite (or at most polynomial-size domain) deterministic functions only as the commitment functionality. Their results showed that ``non-trivial'' functions can not be computed securely by  a black-box use of one-way functions.

Haitner, Omri, and Zarosim~\cite{HaitnerOZ12} studied input-less randomized functionalities and showed that a random oracle\footnote{\cite{HaitnerOZ12} proved this result for a larger class of oracles, see \cite{HaitnerOZ12} for more details.} is, to a large extent, useless for such functionalities as well. In particular, it was shown that for every protocol $\Pi$ in the random oracle model, and every polynomial $p(\cdot)$, there is a protocol in the no-oracle model that is ``$\nicefrac{1}{p(\cdot)}$-close'' to $\Pi$.~\cite{HaitnerOZ12} proved this result by using  the machinery developed in the original publication of this work  (\eg the \emph{graph characterization} of Section~\ref{sec:GraphChar}) and simplified some of the steps of the original proof. \cite{HaitnerOZ12} showed how to use such lower-bounds for the input-less setting to prove black-box separations from one-way functions for ``differentially private'' two-party functionalities for the \emph{with-input} setting. 

\subsection{Our Techniques} \label{sec:techniques}

The main technical challenge in proving our main result is the issue of \emph{dependence} between the executions of the
two parties Alice and Bob in a key agreement protocol.  At first sight, it may seem that a computationally unbounded
attacker that monitors all communication between Alice and Bob will trivially be able to find out their shared key. But
the  presence of the random oracle allows Alice and Bob to correlate their executions even without communicating (which
is indeed the reason that Merkle's protocol achieves nontrivial security). Dealing with such correlations is the cause
of the technical complexity in both our work and the previous work of Impagliazzo and Rudich~\cite{ImpagliazzoRu89}. We
handle this issue in a different way than~\cite{ImpagliazzoRu89}. On a very  high level our approach can be viewed
as using more information about the structure of these correlations than~\cite{ImpagliazzoRu89} did. This allows us to
analyze a more efficient attacking algorithm that is more frugal with the number of queries it uses than the attacker
of~\cite{ImpagliazzoRu89}. Below we provide a more detailed (though still high level) exposition of our technique and
its relation to~\cite{ImpagliazzoRu89}'s technique.


We now review~\cite{ImpagliazzoRu89}'s  attack (and its analysis) and particularly discuss the subtle issue of
\emph{dependence} between Alice and Bob that arises in both their work and ours.
However, no result of this section is used in the later sections, and so the reader should feel free at
any time to skip ahead to the next sections that contain our actual attack and its
analysis.

\subsubsection{The Approach of~\cite{ImpagliazzoRu89}}

Consider a protocol that consists of $n$ rounds of interaction, where each party makes exactly one oracle query before
sending its message.~\cite{ImpagliazzoRu89} called protocols of this type ``normal-form protocols'' and gave an
$\wt{O}(n^3)$ attack against them (their final result was obtained by transforming every protocol into a normal-form
protocol with a quadratic loss of efficiency). Even though without loss of generality the attacker Eve of a key agreement protocol can defer all of her computation till after the interaction between Alice and Bob is finished, it is
conceptually simpler in both~\cite{ImpagliazzoRu89}'s case and ours to think of the attacker Eve as running
concurrently with Alice and Bob. In particular, the attacker Eve of~\cite{ImpagliazzoRu89} performed the following
operations after each round $i$ of the protocol:

\begin{itemize}

\item If the round $i$ is one in which Bob sent a message, then at this point Eve samples  $1000n\log n$ random
    executions of Bob from the distribution $\cD$ of Bob's executions that are consistent with the information that
    Eve has at that moment (which consists of the communication transcript and previous oracle answers). That is, Eve samples a uniformly
    random tape for Bob and uniformly random query answers subject to being consistent with Eve's information.
    After each time she samples an execution, Eve asks the oracle all the queries that are asked during this execution
    and records the answers. (Generally, the true answers will be different from Eve's guessed answers when
    sampling the execution.) If the round $i$ is one in which Alice sent a message, then Eve does similarly by changing the role of Alice and Bob.

\end{itemize}

Overall Eve will sample $\wt{O}(n^2)$ executions making a total of $\wt{O}(n^3)$ queries. It's not hard to see
that as long as Eve learns all of the \emph{intersection queries} (queries asked by both Alice and Bob during the
execution) then she can recover the shared secret with high probability. Thus
the bulk of~\cite{ImpagliazzoRu89}'s analysis was devoted to showing the following claim.

\begin{claim} \label{clm:IR}
With probability at least $0.9$ Eve never fails, where we say that Eve \emph{fails} at round $i$ if
the query made in this round by, say, Alice was asked previously by Bob but not by Eve.
\end{claim}

At first look, it may seem that one could easily prove Claim~\ref{clm:IR}. Indeed, Claim~\ref{clm:IR} will follow by showing that
at any round $i$, the probability that Eve fails in round $i$ \emph{for the first time} is at most  $1/(10n)$. Now all
the communication between Alice and Bob is observed by Eve, and if no failure has yet happened then Eve has also
observed all the intersection queries so far. Because the answers for non-intersection queries are completely random
and independent from one another it seems that Alice has no more information about Bob than Eve does, and hence if the
probability that Alice's query $q$ was asked before by Bob is more than $1/(10n)$ then this query $q$ has probability
at least $1/(10n)$ to appear in each one of Eve's sampled executions of Bob.  Since Eve makes $1000n\log n$ such
samples, the probability that Eve misses $q$ would be bounded by $(1-\tfrac{1}{10n})^{1000n\log n} \ll 1/(10n)$.

\paragraph{The Dependency Issue.}
When trying to turn the above intuition into a proof, the assumption that Eve has as much information about Bob as Alice
does translates to the following statement: conditioned on Eve's information, the distributions of Alice's view and
Bob's view are \emph{independent} from one another.\footnote{Readers familiar with the setting of communication
complexity may note that this is analogous to the well known fact that conditioning on any transcript of a 2-party
communication protocol results in a product distribution (i.e., combinatorial rectangle) over the inputs. However,
things are different in the presence of a random oracle.} Indeed, if this statement were true then the above paragraph
could have been easily translated into a proof that~\cite{ImpagliazzoRu89}'s attacker is successful, and it wouldn't have been
hard to optimize this attacker to achieve $O(n^2)$ queries.  Alas, this statement is false. Intuitively the reason is
the following: even the fact that Eve has not missed any intersection queries is some nontrivial information that
Alice and Bob share and creates dependence between them.\footnote{As a simple example for such dependence consider a
protocol where in the first round Alice chooses $x$ (which is going to be the shared key) to be either the string $0^n$ or $1^n$ at random, queries the
oracle $H$ at $x$ and sends $y=H(x)$ to Bob. Bob then makes the query $1^n$ and gets $y'=H(1^n)$. Now even if Alice
chose $x=0^n$ and hence Alice and Bob have no intersection queries, Bob can find out the value of $x$ just by observing
that $y'\neq y$. Still, an attacker must ask a non-intersection query such as  $1^n$ to know if $x=0^n$ or $x=1^n$.}

Impagliazzo and Rudich~\cite{ImpagliazzoRu89} dealt with this issue by a ``charging argument'', where they showed that such dependence can be charged in a certain way to one of the
executions sampled by Eve, in a way that at most $n$ samples can be charged at each round (and the rest of Eve's
samples are distributed correctly as if the independence assumption was true). This argument inherently required
sampling at least $n$ executions (each of $n$ queries) per round,  resulting in an $\Omega(n^3)$ attack.

\subsubsection{Our Approach}

We now describe our approach and how it differs from the previous proof of~\cite{ImpagliazzoRu89}. The discussion below
is somewhat high level and vague, and glosses over some important details. Again, the reader is welcome to skip ahead
at any time to Section~\ref{sec:desc} that contains the full description of our attack and does not depend on this
section in any way.
Our attacking algorithm follows the same general outline as that of~\cite{ImpagliazzoRu89} but has two important differences:

\begin{enumerate}

\item One \emph{quantitative} difference is that while our attacker Eve also computes a distribution $\cD$ of
    possible executions of Alice and Bob conditioned on her knowledge, she does \emph{not} sample  full
    executions from $\cD$; rather, she computes whether there is any
    query $q\in\bits^*$ that has probability more than, say, $1/(100n)$ of being in $\cD$ and makes
    only such \emph{heavy} queries.

    Intuitively, since Alice and Bob make at most $2n$ queries, the total expected number of heavy queries (and
    hence the query complexity of Eve) is bounded by $O(n^2)$. The actual analysis is more involved since the
    distribution $\cD$ keeps changing as Eve learns more information through the messages she observes and oracle
    answers she receives. 

\item The \emph{qualitative} difference is that here we do not consider the same distribution
    $\cD$ that was considered by~\cite{ImpagliazzoRu89}. Their attacker to some extent ``pretended'' that the
    conditional distributions of Alice and Bob are independent from one another and only considered one party in each round. In contrast, we define our
    distribution $\cD$ to be the \emph{joint} distribution of Alice and Bob, where there could be dependencies
    between them. Thus, to sample from our distribution $\cD$ one would need to sample a \emph{pair} of executions
    of Alice and Bob (random tapes and oracle answers) that are \emph{consistent} with one another and
    Eve's current knowledge.

\end{enumerate}

The main challenge in the analysis is to prove that the attack is \emph{successful} (\ie that
Claim~\ref{clm:IR} above holds) and in particular that the probability of failure at each round (or more generally, at each
query of Alice or Bob) is bounded by, say, $1/(10n)$. Once again, things would have been easy if we knew that the
distribution $\cD$ of the possible executions of Alice and Bob conditioned on Eve's knowledge is a \emph{product distribution}, and hence Alice has no more information on Bob than Eve has. While this is not
generally true, we show that in our attack this distribution is \emph{close to being a product distribution}, in a
precise sense.

At any point in the execution, fix Eve's current information about the system and define a bipartite graph $G$ whose
left-side vertices correspond to possible executions of Alice that are consistent with Eve's information and right-side
vertices correspond to possible executions of Bob consistent with Eve's information. We put an edge between two
executions $A$ and $B$ if they are consistent with one another and moreover if they do not represent an execution in
which Eve has already \emph{failed}  (i.e., there is no intersection query that is asked in both executions $A$ and $B$ but not by Eve).
Roughly speaking, the distribution $\cD$ that our attacker Eve considers can be thought of as choosing a uniformly random
edge in the graph $G$. (Note that the graph $G$ and the distribution $\cD$ change at each point that Eve learns some
new information about the system.) If $G$ were the complete bipartite clique then $\cD$ would have been a product distribution.
Although $G$ can rarely be the complete graph, what we show is that $G$ is still \emph{dense} in the sense that each vertex is connected to most of the vertices on
the other side. Relying on the density of this graph, we show that Alice's probability of hitting a query that Bob asked before is at most
twice the probability that Eve does so if she chooses the most likely query based on her knowledge.

The bound on the degree is obtained by showing that $G$ can be represented as a \emph{disjointness graph}, where each
vertex $u$ is associated with a set $S(u)$ (from an arbitrarily large universe) and there is an edge between a
left-side vertex $u$ and a right-side vertex $v$ if and only if $S(u) \cap S(v) = \es$. The set $S(u)$
 corresponds to the queries made in the execution corresponding to $u$ that are \emph{not} asked by Eve.
The definition of the graph $G$ implies that $|S(u)| \leq n$ for all
vertices $u$. The definition of our attacking algorithm implies that the distribution obtained by picking a random edge
$e=(u,v)$ and outputting $S(u) \cup S(v)$ is \emph{light} in the sense that there is no element $q$ in the universe
that has probability more than $1/(10n)$ of being in a set chosen from this distribution. We show that these
conditions  imply that each vertex is connected to most of the vertices on the other side.

\section{Preliminaries} \label{sec:prelims}

We use bold fonts to denote random variables. By $Q \gets \bQ$ we indicate that $Q$ is sampled from the distribution of the random variable $\bfQ$.
By $(\bfx,\bfy)$ we denote a \emph{joint} distribution over random variables $\bfx,\bfy$. By $\bfx \equiv \bfy$ we denote that $\bfx$ and $\bfy$ are identically distributed. For jointly distributed $(\bfx,\bfy)$, by $(\bfx \mid \bfy=y)$ we denote the distribution of $\bfx$ conditioned on $\bfy=y$. When it is clear from the context we might simply write $(\bfx \mid y)$ instead of $(\bfx \mid \bfy=y)$. By $(\bfx \times \bfy)$ we denote a product distribution in which $\bfx$ and $\bfy$ are sampled independently. For a finite set $S$, by $x \gets S$ we denote that $x$ is sampled from $S$ uniformly at random. By $\Supp(\bfx)$ we denote the \emph{support set} of the random variable $\bfx$ defined as $\Supp(\bfx) = \set{x \mid \Pr[\bfx=x] > 0}$.   For any event $E$, by $\neg E$ we denote the complement  of the event $E$.

\begin{definition}
A \emph{partial function} $F$ is a function $F \colon D \To \bits^*$ defined over some domain $D \se \bits^*$. We call two partial functions $F_1,F_2$ with domains $D_1,D_2$ \emph{consistent} if $F_1(x)=F_2(x)$ for every $x \in D_1 \cap D_2$. (In particular, $F_1$ and $F_2$ are consistent if $D_1 \cap D_2 = \es$.)
\end{definition}

In previous work random oracles are defined either as Boolean functions~\cite{ImpagliazzoRu89} or length-preserving functions~\cite{BellareRo93}. In this work we use a general definition that captures both cases by only requiring the oracle answers to be independent. Since our goal is to give \emph{attacks} in this model, using this definition makes our results more general and applicable to both scenarios.

\begin{definition}[Random Oracles] \label{def:RandOr}
A \emph{random oracle} $\bfH(\cdot)$ is a random variable whose values are functions $H \colon \bits^* \To \bits^*$ such that $\bfH(x)$ is distributed independently of $\bfH(\bits^* \sm \set{x})$ for all $x \in \bits^*$ and that $\Pr[\bfH(x)=y]$ is a rational number for every pair $(x,y)$.\footnote{Our results extend to the case where the probabilities are not necessarily rational numbers, however, since every reasonable candidate random oracle we are aware of satisfies this rationality condition, and it avoids some technical subtleties, we restrict attention to oracles that satisfy it. In Section~\ref{sec:removeRational} we show how to remove this restriction.}  For any finite partial function $F$, by $\Pr_{\bfH}[F]$ we denote the probability that the random oracle $\bfH$ is consistent with $F$. Namely, $\Pr_{\bfH}[F] = \Pr_{H \gets \bfH}[F \se H]$ and $\Pr_{\bfH}[\es] = 1$ where $F \se H$ means that the partial function $F$ is consistent with $H$.
\end{definition}

\begin{remark}[Infinite vs. Finite Random Oracles]
  In this work, we will always work with \emph{finite} random oracles which are only queried on inputs of length  $n \leq \poly(\kappa)$ where $\kappa$ is a (security) parameter given to parties. Thus, we only need a finite variant of Definition \ref{def:RandOr}. However, in case of infinite random oracles (as in Definition \ref{def:RandOr}) we need a measure space over the space of full infinite oracles that is consistent with the finite probability distributions of $\bfH(\cdot)$ restricted to inputs $\bits^n$ for all $n=1,2,\dots$. By Caratheodory's extension theorem, such measure space exists and is unique (see Theorem 4.6 of \cite{ThomasNotes}).
\end{remark}

Since for every random oracle  $\bfH(\cdot)$ and fixed $x$ the random variable $\bfH(x)$ is independent of $\bfH(x')$ for all $x' \neq x$, we can use the following characterization of $\Pr_\bfH[F]$ for every $F \se \bits^* \times \bits^*$. Here we only state and use this lemma for finite sets.

\begin{proposition} \label{prop:ProbOfPartial}
For every random oracle $\bfH(\cdot)$ and every finite set $F \subset \bits^* \times \bits^*$ we have $$\Pr_\bfH[F] = \prod_{(x,y) \in F} \Pr[\bfH(x) = y].$$
\end{proposition}

Now we derive the following lemma from the above proposition.
\begin{lemma} \label{lem:IncExc}
  For consistent finite partial functions $F_1,F_2$ and random oracle $\bfH$ it holds that
  $$\Pr_\bfH[F_1 \cup F_2] = \frac{ \Pr_\bfH[F_1] \cdot \Pr_\bfH[F_2] }{ \Pr_\bfH[F_1 \cap F_2]}.$$
\end{lemma}
\begin{proof} Since $F_1$ and $F_2$ are consistent, we can think of $F = F_1 \cup F_2$ as a partial function. Therefore, by Proposition~\ref{prop:ProbOfPartial} and the  inclusion-exclusion principle we have:
\begin{align*}
\Pr_\bfH[F_1 \cup F_2] &=  \prod_{(x,y) \in F_1 \cup F_2} \Pr[\bfH(x) = y] \\
&= \frac{\prod_{(x,y) \in F_1 } \Pr[\bfH(x) = y] \cdot \prod_{(x,y) \in F_2 } \Pr[\bfH(x) = y]}{\prod_{(x,y) \in F_1 \cap F_2} \Pr[\bfH(x) = y]} \\
&=
\frac{ \Pr_\bfH[F_1] \cdot \Pr_\bfH[F_2] }{ \Pr_\bfH[F_1 \cap F_2]}.
\end{align*}
\end{proof}

\begin{lemma}[Lemma 6.4 in~\cite{ImpagliazzoRu89}] \label{lem:IR-Orig}
Let $E$ be any event defined over a random variable $\bfx$, and let $\bfx_1,\bfx_2,\dots$ be a sequence of random variables all determined by $\bfx$. Let $D$ be the event defined over $(\bfx_1,\dots)$ that holds if and only if there exists some $i \geq 1$ such that $\Pr[E \mid x_1,\dots,x_i] \geq \lambda$. Then $\Pr[E \mid D] \geq \lambda$.
\end{lemma}

\begin{lemma} \label{lem:IR}
Let $E$ be any event defined over a random variable $\bfx$, and let $\bfx_1,\bfx_2,\dots$ be a sequence of random variables all determined by $\bfx$. Suppose $\Pr[E] \leq \lambda$ and $\lambda = \lambda_1 \cdot \lambda_2$. Let $D$ be the event defined over $(\bfx_1,\dots)$ that holds if and only if there exists some $i \geq 1$ such that $\Pr[E \mid x_1,\dots,x_i] \geq \lambda_1$. Then it holds that $\Pr[D] \leq \lambda_2$.
\end{lemma}

\begin{proof}
  Lemma~\ref{lem:IR-Orig} shows that $\Pr[E \mid D] \geq \lambda_1$. Now we prove the contrapositive of Lemma \ref{lem:IR}. If  $\Pr[D ] > \lambda_2$, then we would get $\Pr[E] \geq \Pr[E \land D] \geq \Pr[D] \cdot \Pr[E \mid D] > \lambda_1 \cdot \lambda_2 = \lambda$.
\end{proof}

\subsection{Statistical Distance}

\begin{definition}[Statistical Distance]
By $\Delta(\bfx,\bfy)$ we denote the \emph{statistical distance} between random variables $\bfx,\bfy$ defined as $\Delta(\bfx,\bfy) = \frac{1}{2}\cdot \sum_z|\Pr[\bfx=z] - \Pr[\bfy=z]|$. We call random variables $\bfx$ and $ \bfy$ \emph{$\eps$-close}, denoted by $\bfx \approx_\eps \bfy$,  if $\Delta(\bfx,\bfy) \leq \eps$.
\end{definition}

 We use the following useful well-known lemmas about statistical distance.

\begin{lemma} \label{lem:SDEquivals}
$\Delta(\bfx,\bfy) = \eps$ if and only if either of the following holds:
\begin{enumerate}
  \item For every (even randomized) function $D$ it holds that $\Pr[D(\bfx) =1] - \Pr[D(\bfy)=1] \leq \eps$.
  \item For every event $E$ it holds that  $\Pr_\bfx[E] - \Pr_\bfy[E] \leq \eps$.
\end{enumerate}
Moreover, if $\Delta(\bfx,\bfy) = \eps$, then there is a deterministic (detecting) Boolean function $D$ that achieves $\Pr[D(\bfx) =1] - \Pr[D(\bfy)=1] = \eps$.
\end{lemma}

\begin{lemma} \label{lem:averageSD}
It holds that
$\Delta((\bfx,\bfz),(\bfy,\bfz)) = \Ex_{z \gets \bfz} \Delta((\bfx \mid z),(\bfy \mid z))$.
\end{lemma}

\begin{lemma} \label{lem:triangle}
  If $\Delta(\bfx,\bfy) \leq \eps_1$ and $\Delta(\bfy,\bfz) \leq \eps_2$, then $\Delta(\bfx,\bfz) \leq \eps_1 + \eps_2$.
\end{lemma}

\begin{lemma} \label{lem:project}
 $\Delta((\bfx_1,\bfx_2),(\bfy_1,\bfy_2)) \geq \Delta(\bfx_1,\bfy_1)$.
\end{lemma}

We use the convention for the notation $\Delta(\cdot,\cdot)$ that whenever $\Pr[\bfx \in E]=0$ for some event $E$, we let $\Delta((\bfx \mid E),\bfy)=1$ for every random variable $\bfy$.
%
%
%
%

\begin{lemma} \label{lem:nestedDistance}
Suppose $\bfx,\bfy$ are finite random variables, and suppose $G$ is some event defined over $\Supp(\bfx)$. Then $\Delta(\bfx,\bfy) \leq \Pr_\bfx[G] + \Delta((\bfx \mid \neg G),\bfy)$.
\end{lemma}

\begin{proof}
Let $\delta = \Delta(\bfx,\bfy)$. Let $\bfg$ be a Boolean random variable jointly distributed with $\bfx$  as follows: $\bfg=1$ if and only if $\bfx \in G$. Suppose $\bfy$ is sampled independently of  $(\bfx,\bfg)$ (and so $(\bfy,\bfg) \equiv (\bfy \times \bfg)$). By Lemmas~\ref{lem:project} and~\ref{lem:averageSD} we have:
\begin{align*}
\Delta(\bfx,\bfy)
&\leq \Delta((\bfx,\bfg),(\bfy,\bfg)) \\
&= \Ex_{g \gets \bfg} \Delta((\bfx\mid g),(\bfy \mid g)) \\
&= \Ex_{g \gets \bfg} \Delta((\bfx\mid g),\bfy) \\
&= \Pr[\bfg=1] \cdot \Delta((\bfx\mid \bfg=1),\bfy) + \Pr[\bfg=0] \cdot \Delta((\bfx\mid \bfg=0),\bfy) \\
&\leq \Pr[\bfg=1] + \Delta((\bfx\mid \bfg=0),\bfy) \\
&= \Pr_\bfx[G] + \Delta((\bfx \mid \neg G),\bfy).
\end{align*}
\end{proof}

\begin{definition}[Key Agreement]
  A key agreement protocol consists of two interactive polynomial-time probabilistic Turing machines $(A,B)$ that both get $1^n$ as security parameter, each get secret randomness $\bfr_A,\bfr_B$, and after interacting for $\poly(n)$ rounds $A$ outputs $s_A$ and $B$ outputs $s_B$. We say a key agreement scheme $(A,B)$ has completeness $\rho$ if $\Pr[s_A=s_B] \geq \rho(n)$. For an arbitrary oracle $O$, we  define key agreement protocols (and their completeness) relative to $O$  by simply allowing $A$ and $B$ to be efficient algorithms relative to $O$.
\end{definition}

\paragraph{Security of Key Agreement Protocols.} It can be easily seen that no key agreement protocol with completeness $\rho > 0.9$ could be \emph{statistically} secure, and that there is always a computationally unbounded eavesdropper Eve who can guess the shared secret key $s_A = s_B$ with probability at least $1/2 + \negl(n)$.
In this work we are interested in  statistical security of key agreement protocols in  the \emph{random oracle model}. Namely, we would like to know how many oracle queries are required to break a key agreement protocol relative to a random oracle. 
\section{Proving the Main Theorem} \label{sec:desc}

In this section we prove the next theorem which implies our Theorem~\ref{thm:main} as special case.

\begin{theorem} \label{thm:mainFormal}
Let $\Pi$ be a two-party interactive  protocol between Alice and Bob using a random oracle $\bfH$ (accessible by everyone) such that:

 \begin{itemize}
   \item Alice uses local randomness $r_A$, makes at most $n_A$ queries to $H$ and at the end outputs $s_A$.
   \item Bob uses local randomness $r_B$, makes at most $n_B$ queries to $H$ and at the end outputs $s_B$.
   \item $\Pr[s_A=s_B] \geq \rho$ where the probability is over the choice of $(r_A,r_B,H) \gets (\bfr_A,\bfr_B,\bfH)$.
 \end{itemize}
 Then, for every $0 < \delta < \rho$, there is a \emph{deterministic} eavesdropping adversary Eve who only gets access to the public sequence of messages $M$ sent between Alice and Bob, makes at most $400 \cdot n_A \cdot n_B/\delta^2$ queries to the oracle $H$ and outputs $s_E$ such that $\Pr[s_E = s_B] \geq \rho-\delta$.
\end{theorem}

\subsection{Notation and Definitions} \label{sec:Notation}
In this subsection we give some definitions and notations to be used in the proof of Theorem~\ref{thm:mainFormal}.
\Wolog we assume that Alice, Bob, and Eve will never ask an oracle query twice. Recall that Alice (\resp Bob) asks at most $n_A$ (\resp $n_B$) oracle queries.

\paragraph{Rounds.} Alice sends her messages in odd rounds and Bob sends his messages in even rounds. Suppose  $i=2j-1$ and it is Alice's turn to send the message $m_i$. This round starts by Alice asking her oracle queries and computing $m_i$, then Alice sends  $m_i$ to Bob, and this round ends by Eve asking her (new) oracle queries based on the messages sent so far $M^i=[m_1,\dots,m_i]$. Same holds for $i=2j$ by changing the role of Alice and Bob.

\paragraph{Queries and Views.} By $Q^i_A$ we denote the set of oracle queries asked by Alice by the end of round $i$. By $P^i_A$ we denote the set of oracle query-answer pairs known to Alice by the end of round $i$ (\ie $P^i_A = \set{(q,H(q)) \mid q \in Q^i_A}$).
By $V^i_A$ we denote the view of Alice by the end of round $i$. This view consists of: Alice's randomness $r_A$,   exchanged messages $M^i$ as well as oracle query-answer pairs $P^i_A$ known to Alice  so far. By $Q^i_B,P^i_B,V^i_B$ (\resp $Q^i_E,P^i_E,V^i_E$) we denote the same variables defined for Bob (\resp Eve). Note that $V^i_E$ only consists of $(M^i,P^i_E)$ since Eve does not use any randomness. We also use $\Q(\cdot)$ as an operator that extracts the set of queries from set of query-answer pairs or views; namely, $\Q(P)=\set{q \mid \exists~a, (q,a) \in P}$ and $\Q(V)=\set{q \mid \text{the query } q \text{ is asked in the view }V}$.

\begin{definition}[Heavy Queries]
For a random variable $\bfV$ whose samples $V \gets \bfV$ are sets of queries, sets of query-answer pairs, or views, we say a query $q$ is \emph{$\e$-heavy} for $\bfV$ if and only if $\Pr[q \in \Q(\bfV)] \geq \eps$.
\end{definition}


 \paragraph{Executions and Distributions} A (full) \emph{execution} of Alice, Bob, and Eve can be
described by a  tuple $(r_A,r_B,H)$ where $r_A$ denotes Alice's random tape, $r_B$ denotes Bob's random tape, and $H$
is the random oracle (note that Eve is deterministic). We denote by $\EXEC$ the distribution over (full) executions
that is obtained by running the algorithms for Alice, Bob and Eve with uniformly chosen random tapes $r_A,r_B$ and a uniformly sampled random
oracle $H$. By $\Pr_\EXEC[P^i_A]$ we denote the probability that a full execution of the system leads to $\bfP^i_A=P^i_A$ for a given $P^i_A$. We use the same notation also for other components of the system (by treating their occupance as events) as well.

For a sequence of $i$ messages $M^i=[m_1,\ldots,m_i]$ exchanged between the two parties and  a set of query-answer pairs (\ie a partial function) $P$, by $\Views(M^i,P)$ we denote the joint distribution over the views $(V^i_A,V^i_B)$ of Alice and Bob in their own
(partial) executions up to the  point in the system in which the $i$'th  message is sent (by Alice or Bob) conditioned on: the
transcript of messages in the first $i$ rounds being equal to $M^i$ and $H(q)=a$ for all $(q,a) \in P$. Looking ahead in the proof, the distribution $\Views(M^i,P)$ would be the conditional distribution of Alice's and Bob's views in eyes of the attacker Eve who knows the public messages and has learned oracle query-answer pairs described in $P$.
For $(M^i,P)$ such that $\Pr_\EXEC[M^i,P] >0$, the distribution $\Views(M^i,P)$ can be sampled by first
sampling $(r_A,r_B, H)$ uniformly at random conditioned on being consistent with $(M^i,P)$ and then deriving
Alice's and Bob's views $V^i_A,V^i_B$ from the sampled  $(r_A,r_B, H)$.

For $(M^i,P)$ such that $\Pr_\EXEC[M^i,P] >0$, the event $\Good(M^i,P)$ is defined
over the distribution $ \Views(M^i,P)$ and holds if and only if $Q^i_A \cap Q^i_B \se \Q(P)$ for $Q^i_A, Q^i_B, \Q(P)$ determined by the sampled views $(V^i_A,V^i_B) \gets \Views(M^i,P)$ and $P$.  For  $\Pr_\EXEC[M^i,P] >0$
we define the distribution $\GViews(M^i,P)$ to be the distribution $\Views(M^i,P)$ conditioned on $\Good(M^i,P)$. Looking ahead to the proof the event $\Good(M^i,P)$ indicates that the attacker Eve has not ``missed'' any query that is asked by both of Alice and Bob (i.e. an intersection query) so far, and thus $\GViews(M^i,P)$  refer to the same distribution of $\Views(M^i,P)$ with the extra condition that so far no  intersection query is missed by Eve.

\subsection{Attacker's Algorithm} \label{sec:attackAlg}
In this subsection we describe an attacker Eve who might ask  $\omega(n_A n_B /\delta^2)$ queries, but she finds the key in the two-party key agreement protocol between Alice and Bob with probability $1-O(\delta)$. Then we show how to make Eve ``efficient'' without decreasing the success probability too much.

\paragraph{Protocols in Seminormal Form.} We say a protocol is in \emph{seminormal form}\footnote{We use the term seminormal to distinguish it from the normal form protocols defined in \cite{ImpagliazzoRu89}.} if \num{1} the number of oracle queries asked by Alice or Bob in each round is at most one, and \num{2} when the last message is sent (by Alice or Bob) the other party does not ask any oracle queries and computes its output without using the last message. The second property could be obtained  by simply adding an extra message $\LAST$ at the end of the protocol. (Note that our results do not depend on the number of rounds.)  One can also always achieve the first property  \emph{without} compromising the security as follows. If the protocol has $2\cdot\ell$ rounds, we will increase the number of rounds to $2 \ell \cdot (n_A+n_B-1)$ as follows. Suppose it is Alice's turn to send $m_i$ and before doing so she needs to ask the queries $q_1,\dots,q_k$ (perhaps adaptively) from the oracle. Instead of asking these queries from $H(\cdot)$ and sending $m_i$ in one round, Alice and Bob will run $2n_A-1$ \emph{sub-rounds} of interaction so that Alice will have enough number of (fake) rounds to ask her queries from $H(\cdot)$ one by one. More formally:
\begin{enumerate}
  \item The messages of the first $2n_A-1$ sub-rounds for an odd round $i$ will all be equal to $\bot$. Alice sends the first $\bot$ message, and the last message will be $m_i$ sent by Alice.
  \item For $j \leq k$, before sending the message of the $2j-1$'th sub-round Alice asks $q_j$ from the oracle. The number of these queries, namely $k$, might not be known to Alice at the beginning of round $i$, but since $k \leq n_A$, the number of sub-rounds are enough to let Alice  ask all of her queries $q_1,\dots,q_k$ without asking more than one query in each  sub-round.
\end{enumerate}

If a protocol is in semi-normal form, then in each round there is at most one query asked by the party who sends the message of that round, and we will use this condition in our analysis.
Moreover, Eve can simply \emph{pretend} that \emph{any} protocol is in seminormal form by \emph{imagining} in her head that the extra $\bot$ messages are being sent between every two real message. Therefore, \wolog in the following we will assume that the two-party protocol $\Pi$ has $\ell$ rounds and is in seminormal form.\footnote{Impagliazzo and Rudich~\cite{ImpagliazzoRu89} use the term \emph{normal form} for protocols in which each party  asks \emph{exactly one} query before sending their messages in every round.} Finally note that we cannot simply ``expand'' a round $i$ in which Alice asks $k_i$ queries into $2k$ messages between Alice and Bob, because then Bob would know how many queries were asked by Alice, but if we do the transformation as described above, then the actual number of queries asked for that round could potentially remain secret.

\begin{construction} \label{const:Eve} Let $\eps<1/10$ be an input parameter. The adversary Eve attacks  the $\ell$-round two-party protocol $\Pi$ between Alice and Bob (which is in seminormal form) as follows. During the attack Eve updates a set $P_E$ of oracle  query-answer pairs as follows. Suppose in round $i$ Alice or Bob sends the message $m_i$. After $m_i$ is sent,  if $\Pr_\EXEC[\Good(M^i,P_E)]=0$ holds at any moment, then Eve aborts. Otherwise, as long as there is any query $q \nin \Q(P_E)$ such that
$$\Pr_{(V^i_A,V^i_B) \gets \GViews(M^i,P_E)}[q \in \Q(V^i_A)] \geq \frac{\eps}{n_B} \text{~~~~or~~~~} \Pr_{(V^i_A,V^i_B) \gets \GViews(M^i,P_E)}[q \in \Q(V^i_B)] \geq \frac{\eps}{n_A}$$
(\ie $q$ is  $(\eps/n_B)$-heavy  for Alice or $(\eps/n_A)$-heavy  for Bob with respect to the distribution $\GViews(M^i,P_E)$)
Eve asks the lexicographically first such $q$ from $H(\cdot)$, and adds $(q,H(q))$ to $P_E$.  At the end of round $\ell$ (when Eve is also done with asking her oracle queries), Eve samples $(V'_A,V'_B) \gets \GViews(M^\ell,P^\ell_E)$ and outputs Alice's output $s'_A$ determined by $V'_A$ as its own output $s_E$.
\end{construction}

Theorem~\ref{thm:mainFormal} directly follows from the next two lemmas.

\begin{lemma}[Eve Finds the Key] \label{lem:EveFinds}
  The output $s_E$ of  Eve of Construction~\ref{const:Eve} agrees with $s_B$ with probability at least $\rho-10\eps$ over the choice of $(r_A,r_B,H)$.
\end{lemma}

\begin{lemma}[Efficiency of Eve] \label{lem:EveEff}
The probability that Eve of Construction~\ref{const:Eve} asks more than $n_A \cdot n_B / \eps^2$ oracle queries is at most $10\eps$.
\end{lemma}

Before proving Lemmas~\ref{lem:EveFinds} and~\ref{lem:EveEff} we first derive Theorem~\ref{thm:mainFormal} from them.

\paragraph{Proof of Theorem~\ref{thm:mainFormal}.}  Suppose we modify the adversary Eve and abort it as soon as it asks more than $n_A \cdot n_B / \eps^2$ queries and call the new  adversary EffEve. By Lemmas~\ref{lem:EveFinds} and~\ref{lem:EveEff} the output $s_E$  of
EffEve  still agrees with Bob's output $s_B$ with probability at least $\rho-10\eps-10\eps=\rho-20\eps$. Theorem~\ref{thm:mainFormal} follows by using $\eps=\delta/20 < 1/10$ and noting that $n_A \cdot n_B / (\delta/20)^2 = 400 \cdot n_A \cdot n_B /\delta^2$. \qed

\subsection{Analysis of Attack} \label{sec:analysis}
In this subsection we will prove  Lemmas~\ref{lem:EveFinds} and~\ref{lem:EveEff}, but before doing so  we need some definitions.

\paragraph{Events over $\EXEC$.} Event $\Good$ holds if and only if $Q^\ell_A \cap Q^\ell_B \se Q^\ell_E$ in which case we say that Eve has found all the \emph{intersection queries}. Event $\Fail$ holds if and only if at \emph{some} point during the execution of the system, Alice or Bob asks a query $q$, which was asked by the other party, but not already asked by Eve. If the first query $q$ that makes $\Fail$ happen is Bob's $j$'th query we say the event $\BFail_j$ has happened, and if it is Alice's $j$'th query we say that the event $\AFail_j$ has happened. Therefore, $\BFail_1,\dots,\BFail_{n_B}$ and $\AFail_1,\dots,\AFail_{n_B}$ are disjoint events whose union is equal to $\Fail$. Also note that $\neg \Good \Rightarrow \Fail$, because if Alice and Bob share a query that Eve never made, this must have happened \emph{for the first time} at some point during the execution of the protocol (making $\Fail$ happen), but also note that $\Good$ and $\Fail$ are not necessarily complement events in general. Finally let the event $\BGood_j$ (\resp $\AGood_j$) be the event that when Bob (\resp Alice) asks his (\resp her) $j$'th oracle query, and this happens in round $i+1$, it holds that $Q^i_A \cap Q^i_B \se Q^i_E$. Note that the event $\BFail_i$ implies $\BGood_i$ because if $\BGood_i$ does not hold, it means that Alice and Bob have \emph{already} had an intersection query out of Eve's queries, and so $\BFail_i$ could not be the \emph{first} time that Eve is missing an intersection query.

The following lemma plays a central role in proving both of Lemmas~\ref{lem:EveEff} and~\ref{lem:EveFinds}.

\begin{lemma}[Eve Finds the Intersection Queries]  \label{lem:success}
For all $i \in [n_B]$, $\Pr_{\EXEC}[ \BFail_i ] \leq \frac{3\e}{2n_B}$. Similarly, for all $i \in [n_A]$, $\Pr_{\EXEC}[ \AFail_i ] \leq \frac{3\e}{2n_A}$. Therefore, by a union bound, $\Pr_{\EXEC}[\neg \Good] \leq \Pr_{\EXEC}[\Fail] \leq 3\e$.
\end{lemma}

We will first prove Lemma~\ref{lem:success} and then will use this lemma to prove Lemmas~\ref{lem:EveEff} and~\ref{lem:EveFinds}.
In order to prove Lemma~\ref{lem:success} itself, we will reduce it to stronger statements in two steps \ie Lemmas~\ref{lem:success2} and \ref{lem:combChar}. Lemma~\ref{lem:combChar} (called the graph characterization lemma) is indeed at the heart of our proof and characterizes the conditional distribution of the views of Alice and Bob conditioned on Eve's view.

\subsubsection{Eve Finds Intersection Queries: Proving Lemma~\ref{lem:success}} \label{sec:success}

As we will show shortly, Lemma~\ref{lem:success} follows from the following  stronger lemma.

\begin{lemma} \label{lem:success2}
Let $B_i$, $M_i$, and $P_i$ denote, in order, Bob's view, the sequence of messages sent between Alice and Bob, and the oracle query-answer pairs known to  Eve, all before the moment that Bob is going to ask his $i$'th oracle query that might happen be in a round $j$ that is \emph{different} from $\geq i$.\footnote{Also note that $M_i$ is not necessarily the same as $M^i$. The latter refers to the transcript till the $i$'th message of the protocol is sent, while the former refers to the messages till Bob is going to ask his $i$'th messages (and might ask zero or more than one queries in some rounds).} Then, for \emph{every} $(B_i,M_i,P_i) \gets (\bB_i,\bM_i,\bP_i)$  sampled by executing the system it holds that
$$
\Pr_{\GViews(M_i,P_i)}[ \BFail_i  \mid  B_i ] \leq \frac{3\e}{2n_B}.
$$
A  symmetric statement holds for Alice.
\end{lemma}

We first see why Lemma~\ref{lem:success2} implies Lemma~\ref{lem:success}.

\begin{proof}[Proof of Lemma~\ref{lem:success} using Lemma~\ref{lem:success2}.]
It holds that
$$\Pr[\BFail_i] = \sum_{(B_i,M_i,P_i) \in \Supp(\bB_i,\bM_i,\bP_i)} \Pr_{\EXEC}[B_i,M_i,P_i] \cdot \Pr_{\EXEC}[\BFail_i \mid B_i,M_i,P_i].$$ Recall that as we said the event $\BFail_i$ implies $\BGood_i$. Therefore, it holds that
$$\Pr_{\EXEC}[\BFail_i \mid B_i,M_i,P_i] \leq \Pr_{\EXEC}[\BFail_i \mid B_i,M_i,P_i,\BGood_i]$$
and by definition we have $\Pr_{\EXEC}[\BFail_i \mid B_i,M_i,P_i,\BGood_i] = \Pr_{\GViews(M_i,P_i)}[\BFail_i \mid B_i]$. By Lemma~\ref{lem:success2}  it holds that $\Pr_{\GViews(M_i,P_i)}[\BFail_i \mid B_i]\leq \frac{3\e}{2n_B}$, and so: $$\Pr_{\EXEC}[\BFail_i] \leq \sum_{(B_i,M_i,P_i) \in \Supp(\bB_i,\bM_i,\bP_i)} \Pr_{\EXEC}[B_i,M_i,P_i] \cdot  \frac{3\e}{2n_B} = \Pr[\text{Bob asks  $\geq i$ queries}]\cdot \frac{3\e}{2n_B} \leq \frac{3\e}{2n_B} . $$

\end{proof}

In the following we will prove Lemma~\ref{lem:success2}.  In fact, we will not  use the fact that Bob is about to ask his $i$'th query and will prove a more general statement. For simplicity we  will use a simplified  notation $M=M_i, P=P_i$. Suppose $M=M^j$ (namely the number of messages in $M$ is $j$).
The following graph characterization of the distribution $\Views(M,P)$ is at the heart of our analysis of the attacker Eve of Construction~\ref{const:Eve}. We first describe the intuition and purpose behind the lemma.

\paragraph{Intuition.} Lemma \ref{lem:combChar} below, intuitively, asserts that at any time during the execution of the protocol, while Eve is running her attack, the following holds. Let $(M,P)$ be the view of Eve at any moment. Then the distribution $\Views(M,P)$ of Alice's and Bob's views conditioned on $(M,P)$ could be sampled using a ``labeled'' bipartite graph $G$ by sampling a uniform edge $e = (u,v)$ and taking the two labels of these two nodes (denoted by $A_u,B_v$). This graph $G$ has the extra property of being ``dense'' and close to being a complete bipartite graph.

\begin{lemma}[Graph Characterization of $\Views(M,P)$] \label{lem:combChar}
Let $M$ be the sequence of messages sent between Alice and Bob, let $P$ be the set of oracle query-answer pairs known to Eve by the end of the round in which the last message in  $M$ is sent and Eve is also done with her learning queries. Let $\Pr_{\Views(M,P)}[\Good(M,P)]>0$. For \emph{every} such $(M,P)$, there is a bipartite graph $G$ (depending on $M,P$) with vertices $(\cU_A,\cU_B)$ and edges $E$ such that:
\begin{enumerate}
  \item Every vertex $u$ in $\cU_A$ has a corresponding view $A_u$ for Alice (which is consistent with $(M, P)$) and a set $Q_u = \Q(A_u) \sm \Q(P)$, and the same holds for vertices in $\cU_B$ by changing the role of Alice and Bob. (Note that every view can have multiple vertices assigned to it.)
  \item \label{item:Disjoint} There is an edge between $u \in \cU_A$ and $v \in \cU_B$ if and only if $Q_u \cap Q_v = \es$.
  \item \label{item:Degrees} Every vertex  is connected to at least a $(1-2\eps)$ fraction of the vertices in the other side.

  \item \label{item:EquivDists} The distribution $(V_A,V_B) \gets \GViews(M,P)$ is identical to: sampling a random edge $(u,v) \gets E$ and taking $(A_u,B_v)$ (\ie the views corresponding to $u$ and $v$).
  \item \label{item:SameSupp} The distributions $\GViews(M,P)$ and $\Views(M,P)$ have the same support set.
\end{enumerate}
\end{lemma}

Lemma~\ref{lem:combChar} at  the heart of the proof of our main theorem, and so we will first see how to use this lemma before proving it. In particular, we first use Lemma~\ref{lem:combChar} to prove Lemma~\ref{lem:success2}, and then we will prove Lemma~\ref{lem:combChar}.

\paragraph{Proof of Lemma~\ref{lem:success2} using Lemma \ref{lem:combChar}.} Let
$B=B_i,M=M_i,P=P_i$ be as in Lemma~\ref{lem:success2} and $q$ be Bob's $i$'th query which is going to be asked after the last message $m_j$ in $M=M_i=M^j$ is sent to Bob. By
Lemma~\ref{lem:combChar}, the distribution $\GViews(M,P)$ conditioned on getting $B$ as Bob's view is the same as uniformly sampling a random edge $(u,v) \gets E$ in the graph $G$ of Lemma~\ref{lem:combChar} conditioned on $B_v=B$.
We prove Lemma~\ref{lem:success2} even conditioned on choosing any vertex $v$ such that $B_v = B$. For such fixed $v$,  the distribution of Alice's view $A_u$, when we choose a random edge $(u,v')$ conditioned on $v=v'$ is the same as choosing a random neighbor $u \gets N(v)$ of the node $v$ and then selecting Alice's view $A_u$ corresponding to the node $u$. Let $S = \{u
\in \cU_A \text{ such that } q \in A_u \}$. Assuming $d(u)$ denotes the degree of  $w$ for any node $w$ we have
\[
\Pr_{u \gets N(v)}[q \in A_u] \leq \frac{\abs{S}}{d(v)} \leq \frac{\abs{S}}{(1-2\e) \cdot |\cU_A|} \leq \frac{\abs{S} \cdot |\cU_B|}{(1-2\e) \cdot \abs{E}} \leq
\frac{\sum_{u \in S} d(u)}{(1-2\e)^2 \cdot \abs{E}} \leq \frac{\e}{(1-2\e)^2 \cdot n_B} < \frac{3\e}{2n_B}.
\]
First note that proving the above inequality is sufficient for the proof of Lemma~\ref{lem:success2}, because $\BFail_i$ is equivalent to $q \in A_u$. Now, we prove the above inequalities.

The second and fourth inequalities are due to the degree lower bounds of Item~\ref{item:Degrees} in Lemma~\ref{lem:combChar}. The third inequality is because $\abs{E} \leq
\abs{\cU_A} \cdot \abs{\cU_B}$. The fifth inequality is because of the definition of the attacker Eve who asks $\e/n_B$ heavy queries for Alice's view when sampled from $\GViews(M,P)$, as long as such queries exist. Namely, when we choose a random edge $(u,v) \gets E$ (which by Item~\ref{item:EquivDists} of Lemma~\ref{lem:combChar} is the same as sampling $(V_A,V_B) \gets \GViews(M,P)$), it holds that $u \in S$ with probability $\sum_{u \in S} d(u) / |E|$.  But for all $u \in S$ it holds that $q \in Q_u$, and so if $\sum_{u \in S} d(u) / |E| > \eps/n_B$ the query $q$ should have been learned by Eve already and so  $q$ could not be in any set $Q_u$.
The sixth inequality is because we are assuming $\e < 1/10$. \qed

\subsubsection{The Graph Characterization: Proving Lemma~\ref{lem:combChar}} \label{sec:GraphChar}
We prove Lemma~\ref{lem:combChar} by first presenting a ``product characterization'' of the distribution $\GViews(M,P)$.\footnote{A similar observation was made by~\cite{ImpagliazzoRu89}, see Lemma~6.5 there.}
%

%

\begin{lemma}[Product Characterization]
\label{lem:product} For any $(M,P)$ as described in Lemma~\ref{lem:combChar} there exists a distribution $\bfA$ (\resp $\bfB$) over Alice's
 (\resp Bob's) views such that the distribution $\GViews(M,P)$ is identical to the product distribution $(\bfA \times \bfB)$ conditioned on the event $\Good(M,P)$. Namely,

\[
\GViews(M,P) \equiv ((\bfA \times \bfB) \mid \Q(\bfA) \cap \Q(\bfB) \se \Q(P)).
\]

\end{lemma}

\begin{proof}
Suppose $(V_A,V_B) \gets \Views(M,P)$ is such that $Q_A \cap Q_B \se Q$ where $Q_A = \Q(V_A), Q_B = \Q(V_B)$, and $Q = \Q(P)$. For such $(V_A,V_B)$ we will show that  $\Pr_{\GViews(M,P)}[(V_A,V_B)] = \alpha(M,P) \cdot \alpha_A \cdot \alpha_B$ where: $\alpha(M,P)$ only depends on $(M,P)$, $\alpha_A$  only depends on
$V_A$, and $\alpha_B$ only depends only on $V_B$. This means that if we let $\bfA$ be the
distribution over $\Supp(V_A)$ such that $\Pr_{\bfA}[V_A]$ is proportional to $\alpha_A$ and let $\bfB$ be the distribution  over $\Supp(V_B)$  such that
$\Pr_{\bfB}[V_B]$ is proportional to $\alpha_B$, then $\GViews(M,P)$ is proportional (and hence equal to) the
distribution  $((\bfA \times \bfB) \mid Q_A \cap Q_B \se Q)$.

In the following we will show that $\Pr_{\GViews(M,P)}[(V_A,V_B)] = \alpha(M,P) \cdot \alpha_A \cdot \alpha_B$. Since we are assuming $Q_A \cap Q_B \se Q$ (\ie that the event $\Good(M,P)$ holds over $(V_A,V_B)$) we have:
\begin{equation} \label{eq:1}
\Pr_{\Views(M,P)}[(V_A,V_B)] = \Pr_{\Views(M,P)}[(V_A,V_B) \wedge \Good(M,P)] =  \Pr_{\Views(M,P)}[ \Good(M,P) ]
\Pr_{\GViews(M,P)}[(V_A,V_B)].
\end{equation}

On the other hand, by definition of conditional probability we have\footnote{Note that $V_A, V_B$ uniquely determine $M, P$ so  $\Pr[V_A, V_B, M, P] = \Pr[V_A,V_B]$ holds for consistent $V_A, V_B, M, P$, but we choose to write the full event's description for clarity.}
\begin{equation} \label{eq:2}
  \Pr_{\Views(M,P)}[(V_A,V_B)] =
\frac{\Pr_{\EXEC}[(V_A,V_B,M,P) ]}{\Pr_{\EXEC}[(M,P) ]}.
\end{equation}

Therefore, by Equations (\ref{eq:1}) and (\ref{eq:2}) we have

\begin{equation} \label{eq:3}
\Pr_{\GViews(M,P)}[(V_A,V_B)] = \frac{\Pr_{\EXEC}[(V_A,V_B,M,P) ]}{\Pr_{\EXEC}[(M,P) ] \cdot \Pr_{\Views(M,P)}[ \Good(M,P) ]}.
\end{equation}

The denominator of the righthand side of Equation (\ref{eq:3}) only depends on $(M,P)$ and so we can take $\beta(M,P) = \Pr_{\EXEC}[(M,P) ] \cdot \Pr_{\Views(M,P)}[ \Good(M,P) ]$. In the following we analyze the numerator.

Recall that for a partial function $F$, by $\Pr_\EXEC[F]$ we denote the probability that $H$ from the sampled execution $(r_A,r_B,H) \gets \EXEC$ is consistent with $F$; namely, $\Pr_\EXEC[F] = \Pr_{\bfH}[F]$ (see Definition~\ref{def:RandOr}).

Let $P_A$ (\resp $P_B$) be the set of oracle query-answer pairs in $V_A$ (\resp $V_B$).
 We claim that:

\[
\Pr_{\EXEC}[(V_A,V_B,M,P) ]=
\Pr[\bfr_A=r_A] \cdot \Pr[\bfr_B = r_B] \cdot \Pr_\EXEC[P_A
\cup P_B \cup P].
\]

The reason is that the necessary and sufficient condition that
 $(V_A,V_B,M,P)$ happens in the execution of the system is that when  we sample  a uniform $(r_A, r_B, H)$, $r_A$ equals Alice's randomness,
 $r_B$ equals Bob's randomness, and $H$ is consistent with $P_A
\cup P_B \cup P$. These conditions implicitly imply that Alice and Bob will
indeed produce the transcript $M$ as well.

Now by Lemma~\ref{lem:IncExc} and  $(P_A \cap P_B) \sm P = \es$ we have $\Pr_\EXEC[P_A
\cup P_B \cup P]$ equals to:

$$  \Pr_\EXEC[P] \cdot \Pr_\EXEC[(P_A
\cup P_B ) \sm P] =
\frac{\Pr_\EXEC[P]\cdot \Pr_\EXEC[P_A \sm P] \cdot \Pr_\EXEC[P_B \sm P] }{ \Pr_\EXEC[(P_A \cap P_B) \sm P]}
=\Pr_\EXEC[P]\cdot \Pr_\EXEC[P_A \sm P] \cdot \Pr_\EXEC[P_B \sm P].$$

Therefore, we get:
\[
\Pr_{\GViews(M,P)}[(V_A,V_B)] = \frac{\Pr[\bfr_A=r_A] \cdot \Pr[\bfr_B = r_B] \cdot \Pr_\EXEC[P]  \cdot \Pr_\EXEC[P_A \sm P]  \cdot \Pr_\EXEC[P_B \sm P]}{\beta(M,P)}.
\]
and so we can take $\alpha_A = \Pr[\bfr_A=r_A] \cdot \Pr_\EXEC[P_A \sm P] $, $\alpha_B = \Pr[\bfr_B=r_B] \cdot \Pr_\EXEC[P_B \sm P]$, and $\alpha(M,P) = \Pr_\EXEC[P] / \beta(M,P)$.
\end{proof}

 \paragraph{Graph Characterization.} The product characterization of Lemma~\ref{lem:product} implies that we can think of $\GViews(M,P)$ as a
distribution over random edges of some bipartite graph
$G=(\cU_A,\cU_B,E)$ defined based on $(M,P)$ as follows.

\begin{construction}[Labeled graph $G=(\cU_A,\cU_B,E)$] \label{const:graph} Every node $u \in \cU_A$ will have a corresponding view $A_u$ of Alice that is in the
support of the distribution $\bfA$  from Lemma~\ref{lem:product}.
We also let the number of nodes corresponding to a
view $V_A$ be proportional to $\Pr_{\bfA}[\bfA=V_A]$, meaning that $\bfA$ corresponds to the uniform distribution over the
left-side vertices $\cU_A$.
Similarly, every node $v\in \cU_B$ will have a corresponding view $B_v$  of Bob such that $\bfB$
corresponds to the uniform distribution over $\cU_B$.
Doing this is  possible because the probabilities $\Pr_{\bfA}[\bfA=V_A]$ and  $\Pr_{\bfB}[\bfB=V_B]$  are all rational numbers.
More formally, since in Definition~\ref{def:RandOr} of random oracles we assumed  $\bfH(x)=y$ to be rational for all $(x,y)$,  the probability space $\GViews(M,P)$ only includes rational probabilities. Thus, if $W_1,\dots,W_N$ is the list of all possible views for Alice when sampling $(V_A,V_B) \gets \GViews(M,P)$, and if $\Pr_{(V_A,V_B) \gets \GViews(M,P)} [W_j=V_A] = c_j/d_j$ where $c_1,d_1,\dots,c_N,d_N$ are all integers, we can put
$(c_j/d_j) \cdot \prod_{i \in [N]} {d_i} $
many nodes in $\cU_A$ representing the view $W_j$. Now if we sample a node  $u \gets \cU_A$ uniformly and take $A_u$ as Alice's view, it would be the same as sampling $(V_A,V_B) \gets \GViews(M,P)$ and taking $V_A$.
Finally, we define $Q_u = Q(A_u) \setminus Q(P)$  for $u \in \cU_A$ to be the
set of queries \emph{outside of $\Q(P)$} that were asked by Alice in the view $A_u$. We define $Q_v = Q(B_u) \setminus
Q(P)$ similarly. We put an edge  between the nodes $u$ and $v$ (denoted by $u \sim v$) in $G$ if and only if $Q_u \cap Q_v
= \es$.
\end{construction}

It turns out that the graph $G$
is \emph{dense} as formalized in the next lemma.
%

\begin{lemma} \label{lem:highDeg}
Let $G=(\cU_A,\cU_B,E)$ be the graph of Construction \ref{const:graph}. Then for every $u \in \cU_A, d(u) \geq |\cU_B|\cdot (1-2\e)$ and for every $v \in \cU_B$,
$d(v) \geq |\cU_A|\cdot (1-2\e)$ where $d(w)$ is the degree of the vertex $w$.
\end{lemma}

\begin{proof} First note that Lemma~\ref{lem:product} and the description of Construction \ref{const:graph} imply that the distribution $\GViews(M,P)$ is equal to the distribution
obtained by letting $(u,v)$ be a random edge of the graph $G$ and choosing $(A_u, B_v)$. We will make use of this property.

We first show that for every $w \in \cU_A$, $\sum_{v \in \cU_B, w \not \sim v} d(v) \leq \e \cdot \abs{E}$. The reason is that
the probability of vertex $v$ being chosen when we choose a random edge is $\frac{d(v)}{\abs{E}}$ and if $\sum_{v \in
\cU_B, w \not \sim v} \frac{d(v)}{\abs{E}} > \e$, it means that $\Pr_{(u,v) \gets E}[Q_w \cap Q_v \neq \es] \geq
\e$. Hence, because $\abs{Q_w} \leq n_A$, by the pigeonhole principle there would exist $q \in Q_w$ such that  $\Pr_{(u,v)
\gets E}[q \in Q_v] \geq \e/n_A$. But this is a contradiction, because if that holds, then $q$ should have been in $P$ by the definition of
the attacker Eve of Construction~\ref{const:Eve}, and hence it could not be in $Q_w$. The same argument shows that for every $w \in \cU_B$, $\sum_{u \in \cU_A, u \not
\sim w} d(u) \leq \e \abs{E}$. Thus, for every vertex $w \in \cU_A \cup \cU_B$, $\abs{E^{\not\sim}(w)} \leq \e\abs{E}$ where
$E^{\not\sim}(w)$ denotes the set of edges that do not contain any neighbor of $w$ (i.e., $E^{\not \sim}(w) =
\{(u,v) \in E \mid u \not \sim w \wedge w \not \sim v \}$). The following claim proves Lemma~\ref{lem:highDeg}.

\begin{claim}
For $\e \leq 1/2$, let $G=(\cU_A,\cU_B,E)$ be a nonempty bipartite graph where $\abs{E^{\not \sim}(w) } \leq \e
\abs{E}$ for all vertices $w \in \cU_A \cup \cU_B$. Then $d(u) \geq |\cU_B| \cdot (1-2\e)$  for all $u \in \cU_A$ and $d(v) \geq
|\cU_A| \cdot (1-2\e)$ for all $v \in \cU_B$.
\end{claim}

\begin{proof}
Let $d_A = \min \{ d(u) \mid u \in \cU_A \}$ and $d_B = \min \{ d(v) \mid v \in \cU_B \}$.  By switching the left and right sides if necessary, we may assume without loss of generality that
\begin{equation} \label{eq:4}
\frac{d_A}{\abs{\cU_B}} \leq \frac{d_B}{\abs{\cU_A}}.
\end{equation}
So it suffices to prove that $1-2\e \leq
\frac{d_A}{\abs{\cU_B}}$. Suppose $1-2\e > \frac{d_A}{\abs{\cU_B}}$, and let $u \in \cU_A$ be the vertex that $d(u) = d_A <
(1-2\e) \abs{\cU_B}$. Because for all $v\in \cU_B$ we have $d(v) \leq \abs{\cU_A}$, thus, using Inequality (\ref{eq:4}) we get that $\abs{
E^{\sim} (u)} \leq d_A \abs{\cU_A}   \leq d_B \abs{\cU_B}$ where $E^{\sim} (u) = E \setminus E^{\not \sim} (u) $. On the
other hand since we assumed that $d(u) < (1-2\e)|\cU_B|$, there are more than $2\e|\cU_B|d_B$ edges in $E^{\not\sim}(u)$,
meaning that $\abs{E^{\sim}(u)} < \abs{E^{\not \sim} (u) }/(2\e)$. But this implies
\[
|E^{\not \sim} (u) | \leq \e|E|=\e\left(|E^{\not \sim} (u) |+ |E^{\sim}(u)|\right) <
\e|E^{\not \sim} (u) | +  |E^{\not \sim} (u) |/2 ,
\]
which is a contradiction for $\e<1/2$.
\end{proof}
Finally we prove Item \ref{item:SameSupp}. Namely, for every $(A,B) \gets \Views(\bfV_A,\bfV_B)$, there is some $B'$ such that $(A,B')$ is in the support set of $\GViews(\bfV_A,\bfV_B)$. The latter is equivalent   to finding $B'$ that is consistent with $M,P$ and that $\cQ(A) \cap \cQ(B) \se \cQ(P)$. For sake of contradiction suppose this is not the case. Therefore, if we sample $B'$ from the distribution of $\bfV_B$ conditioned on $P,M$ then there is always an element in $\cQ(A) \cap \cQ(B')$ that is outside of $cQ(P)$. By the pigeonhole principle, one of the queries in $\cQ(A) \sm \cQ(P)$ would be at least $1/n_A$-heavy for the distribution $\GViews(\bfV_A,\bfV_B)$ (in particular the $\bfV_B$ part). But this contradicts how the algorithm of Eve operates.
\end{proof}

\begin{remark}[Sufficient Condition for Graph Characterization] \label{rem:GraphChar}
It can be verified that the proof of the graph characterization of Lemma~\ref{lem:combChar} only requires the following: At the end of the rounds, Eve has learned all the $(\eps/n_B)$-heavy queries for Alice and all the $(\eps/n_A)$-heavy queries for Bob with respect to the distribution $\GViews(M,P)$.
More formally, all we need is that when Eve stops asking more queries, if there is any query $q$ such that
$$\Pr_{(V_A,V_B) \gets \GViews(M,P)}[q \in \Q(V_A)] \geq \frac{\eps}{n_B} \text{~~~~or~~~~} \Pr_{(V_A,V_B) \gets \GViews(M,P_)}[q \in \Q(V_B)] \geq \frac{\eps}{n_A}$$
then $q \in \Q(P)$. In particular,  Lemma~\ref{lem:combChar}   holds even if Eve arbitrarily asks queries that are \emph{not} necessarily heavy at the time being asked or chooses to ask the heavy queries in an arbitrary (different than lexicographic) order.
\end{remark}

\subsubsection{Eve Finds the Key: Proving Lemma~\ref{lem:EveFinds}} \label{sec:findSecret}

Now, we turn to the question of finding the secret. Theorem~6.2 in~\cite{ImpagliazzoRu89} shows that once one finds all
the intersection queries, with $O(n^2)$ more queries they can also find the actual secret. Here we use the properties
of our attack to show that we can do so even without asking more queries.

First we need to specify and prove the following corollary of of Lemma~\ref{lem:combChar}.

\begin{corollary}[Corollary of Lemma~\ref{lem:combChar}] \label{cor:statClose}
Let Eve be the eavesdropping adversary of Construction~\ref{const:Eve} using parameter $\eps$, and $\Pr_{\Views(M^i,P^i_E)}[\Good(M^i,P^i_E)]>0$ where $(M^i,P^i_E)$ is the view of Eve by the end of round $i$ (when she is also done with learning queries). For the fixed $i,M^i,P^i_E$, let $(\bfV_A, \bfV_B)$ be the joint view of Alice and Bob as sampled from $\GViews(M^i,P^i_E)$. Then for some product distribution $(\bfU_A \times \bfU_B)$ (where $\bfU_A \times \bfU_B$ could also depend on  $i,M^i,P^i_E$) we have:
\begin{enumerate}
\item $\Delta((\bfV_A,\bfV_B) , (\bfU_A \times \bfU_B)) \leq 2\eps$.
\item For \emph{every} possible $(A,B) \gets \Views(\bfV_A,\bfV_B)$ (which by Item \ref{item:SameSupp} is the same as the set of all $(A,B) \gets \GViews(\bfV_A,\bfV_B)$) we have:
\begin{align*}
\Delta((\bfV_A \mid \bfV_B=B) , \bfU_A) &\leq 2\eps, \\
\Delta((\bfV_B \mid \bfV_A=A) , \bfU_B) &\leq 2\eps.
\end{align*}
\end{enumerate}
\end{corollary}

\begin{proof}
In the graph characterization $G=(\cU_A,\cU_B,E)$ of $\GViews(M,P)$ as described in Lemma~\ref{lem:combChar}, every vertex is connected to $1-2\eps$ fraction of the vertices of the other section, and consequently the graph $G$ has $1-2\eps$ fraction of  the edges of the complete bipartite graph with the same nodes $(\cU_A,\cU_B)$. Thus, if we take $\bfU_A$ the uniform distribution over $\cU_A$ and $\bfU_B$  the uniform distribution over $\cU_B$, they satisfy all the three inequalities.
\end{proof}

The process of sampling the components of the system can also be done in a ``reversed''  order where we first decide about whether some events are going to hold or not and then sample the other components conditioned on that.

\paragraph{Notation.} In the following let $s(V)$ be the output determined by any view $V$ (of Alice or Bob)

\begin{construction} \label{const:sample} Sample Alice, Bob, and Eve's views as follows.
\begin{compactenum}
  \item Toss a coin $b$ such that $b=1$ with probability $\Pr_\EXEC[\Good]$.

  \item If $b=1$:
    \begin{compactenum}
    \item Sample Eve's final view $(M,P)$ conditioned  on $\Good$.
    \item \label{item:b1}

      \begin{compactenum}
      \item Sample  views of Alice and Bob $(V_A,V_B)$ from $\GViews(M,P)$. 
      \item Eve samples $(V'_A,V'_B) \gets \GViews(M,P)$, and outputs $s_E=s(V'_A)$.
      \end{compactenum}
    \end{compactenum}

  \item If $b=0$:
    \begin{compactenum}
    \item Sample Eve's final view $(M,P)$ conditioned on $\neg \Good$.
    \item \label{item:b0}
      \begin{compactenum}
      \item Sample  views  $(V_A,V_B) \gets (\Views(M,P) \mid \neg\Good)$. 
      \item Eve does the same as case $b=1$ above.
      \end{compactenum}
    \end{compactenum}

\end{compactenum}
\end{construction}

In other words, $b=1$ if and only if $\Good$ holds over the real views of Alice and Bob. We might use $b=1$ and  $\Good$  interchangeably (depending on which one is conceptually more convenient).

The attacker Eve of Construction~\ref{const:Eve}   samples views $(V'_A,V'_B)$ from $\GViews(M,P)$ in \emph{both} cases of $b=0$ and $b=1$, and that is exactly what the Eve of Construction~\ref{const:sample} does as well, and the pair $(s_E,s(V_B))$  in Constructions~\ref{const:Eve} vs.~\ref{const:sample}  are  identically distributed. Therefore, our goal is to  lower bound the probability of getting $s_E= s(V_B)$ where $s_E=s(V'_A)$ is the output of $V'_A$ and $s(V_B)$ is the output of $V_B$ (in Construction \ref{const:sample}). We would show that this event happens in Step~\ref{item:b1} with sufficiently large probability. (Note that it is also possible that $s_E=s(V_B)$ happens in Step~\ref{item:b0} as well, but we ignore this case.)

In the following, let $\rho(M,P)$ and $ \win(M,P)$ be defined as follows.
\begin{align*}
\rho(M,P) &= \Pr_{(V_A,V_B) \gets \GViews(M,P)}[s(V_A)=s(V_B)] \\
\win(M,P) &= \Pr_{(V_A,V_B) \gets \GViews(M,P), (V'_A,V'_B) \gets \GViews(M,P)}[s(V'_A) = s(V_B)]
\end{align*}
 where $(V_A,V_B)$ and $(V'_A,V'_B)$ are independent samples.

We will prove Lemma~\ref{lem:EveFinds} using the following two claims.

\begin{claim} \label{clm:twoSamples}
   Suppose  $P$ denotes Eve's set of oracle query-answer pairs after all of the messages in $M$ are sent. Assuming the probability of $\Good(M, P)$ is nonzero conditioned on $(M,P)$, for every $\eps<1/10$ used by Eve's algorithm of Construction~\ref{const:Eve} it  holds that
  $$ \win(M,P) \geq \rho(M,P) - 4\eps.$$
\end{claim}

Now we prove Claim \ref{clm:twoSamples}.

\begin{proof}[Proof of Claim \ref{clm:twoSamples}]
Let $(\bfU_A \times \bfU_B)$ be the product distribution of Corollary~\ref{cor:statClose} for the view of $(M,P)$.
We would like to lower bound the probability of $s(V'_A)=s(V_B)$ where $(V_A,V_B)$ and $(V'_A,V'_B)$ are independent samples from the same distribution $(\bfV_A,\bfV_B) \equiv \GViews(M,P)$. Since $M,P$ are fixed, for simplicity of notation, in the following we let $(\bfV_A,\bfV_B) \equiv \GViews(M,P)$ without explicitly mentioning $M,P$. Also, in what follows, $\bfV_A$ (\resp $\bfV_B$) will denote the marginal distribution of the first (\resp second) component of  $(\bfV_A,\bfV_B)$. We will also preserve $V_A,V_B$ to denote the real and Bob views sampled from $(\bfV_A,\bfV_B)$, and we will use $V'_A,V'_B$ to denote Eve's samples from the same distribution $(\bfV_A,\bfV_B)$.
 
 For every possible view $A_0 \gets \bfV_A$, let $\rho(A_0)=\Pr_{(A,B) \gets (\bfV_A,\bfV_B))}[s(A)=s(B) \mid A=A_0]$.
By averaging over Alice's view, it holds that $\rho(M,P) = \Ex_{(A,B) \gets (\bfV_A,\bfV_B)} [\rho(A)]$. Similarly, for every possible view $A_0 \gets \bfV_A$, let $\win(A_0)=\Pr_{(A,B) \gets (\bfV_A,\bfV_B))}[s(A)=s(B)]$.
By averaging over Alice's view, it holds that $\rho(M,P) = \Ex_{(A,B) \gets (\bfV_A,\bfV_B)} [\rho(A)]$ and $\win(M,P) = \Ex_{(A,B) \gets (\bfV_A,\bfV_B)} [\win(A)]$

In the following, we will prove something stronger than Claim~\ref{clm:twoSamples} and will show that  $\win(V'_A) \geq \rho(V'_A) - 4\eps$  for \emph{every}  $V'_A \gets \bfV_A$, and the claim follows by averaging over $V'_A \gets \bfV_A$. Thus, in the following $V'_A$ will be the fixed sample $V'_A \gets \bfV_A$.
By Corollary~\ref{cor:statClose}, for every possible Alice's view $A \gets \bfV_A$, the distribution of Bob's view  sampled from $(\bfV_B \mid \bfV_A = A)$ is $2\eps$-close to $\bfU_B$. Therefore, the distribution of $\bfV_B$ (without conditioning on $\bfV_A=A$) is also $2\eps$-close to $\bfU_B$.
By two applications of Lemma~\ref{lem:SDEquivals} we get
\begin{align*}
\win(V'_A) &= \Pr_{V_B \gets \bfV_B}[s(V'_A)=s(V_B)] \\
&\geq \Pr_{B \gets \bfU_B}[s(V'_A)=s(B)] - 2\eps \\
&\geq
\Pr_{V'_B \gets (\bfV_B \mid \bfV_A =V'_A)}[s(V'_A)=s(V'_B)] - 4\eps \\
&= \rho(V'_A) - 4\eps.
\end{align*}
\end{proof}

The following claim lower bounds the completeness of the key agreement protocol when conjuncted with  reaching Step~\ref{item:b1} in Construction~\ref{const:sample}.
\begin{claim} \label{clm:AgreeCondOnGood}
It holds that $\Pr_\EXEC[s(V_A) = s(V_B) \land \Good] \geq \rho - 3\eps$.
\end{claim}
\begin{proof}

By Lemma~\ref{lem:success}  it holds that $1-3 \eps\leq \Pr_\EXEC[\Good] $. Therefore

$$\rho - 3\eps  \leq \Pr_\EXEC[s(V_A) = s(V_B)] - \Pr_\EXEC[\neg \Good] =
 \Pr_\EXEC[s(V_A) = s(V_B) \land \Good] .$$
 

\end{proof}

\paragraph{Proof of Lemma~\ref{lem:EveFinds}.}
We will show a stronger claim that $\Pr[s(V'_A) = s(V_B) \land \Good] \geq \rho-7\eps$ which implies $\Pr[s(V'_A)=s(V_B)] \geq \rho - 7\eps$ as well.
By definition of Construction~\ref{const:sample} and using Claims~\ref{clm:twoSamples} and~\ref{clm:AgreeCondOnGood} we have:
\begin{align*}
\Pr[s(V'_A) = s(V_B) \land \Good]
&= \Pr_\EXEC[\Good] \cdot \Ex_{(M,P) \gets ((\bM,\bP) \mid \Good)} [\win(M,P)] \\
&\geq \Pr_\EXEC[\Good] \cdot \Ex_{(M,P) \gets ((\bM,\bP) \mid \Good)} [\rho(M,P) - 4\eps] \\
&= \big(\Pr_\EXEC[\Good] \cdot \Ex_{(M,P) \gets ((\bM,\bP) \mid \Good)} [\rho(M,P)]\big)- (4 \Pr_\EXEC[\Good] \cdot \eps )\\
&= \big(\Pr_\EXEC[\Good] \cdot \Pr[s(V_A)=s(V_B) \mid \Good]\big) - (4 \Pr_\EXEC[\Good] \cdot \eps)  \\
& \geq (\rho-3\eps)-(4\eps) = \rho-7\eps.
\end{align*}

\qed

\subsubsection{Efficiency of Eve: Proving Lemma~\ref{lem:EveEff}} \label{sec:efficiency}

Recall that Eve's
criteria for ``heaviness''  is based on  the distribution $\GViews(M,P_E)$ where $M$ is the current sequence of messages sent so far and $P_E$ is the current set of oracle query-answer pairs  known to Eve. This distribution is conditioned on Eve not missing any
queries up to this point. However, because we have proven that the event $\Fail$ has small probability,
queries that are heavy under $\GViews(M,P_E)$ are also (typically) almost as heavy under the real distribution $\Views(M,P_E)$. Intuitively this means that, on average, Eve will not make too many queries.

\begin{definition} [Coloring of Eve's Queries] \label{def:color} Suppose $(M^i,P_E)$ is the view of Eve at the moment Eve asks query $q$. We call $q$  a \emph{red} query, denoted $q \in \RED$, if $\Pr[\Good(M^i,P_E)] \leq 1/2$. We call $q$ a \emph{green} query of Alice's type, denoted $q \in \GRA$, if $q$ is not red and $\Pr_{(V^i_A,V^i_B) \gets \Views(M^i,P_E)}[q \in \Q(V^i_A)] \geq \frac{\eps}{2n_B}$. (Note that here we are sampling the views from $\Views(M^i,P_E)$ and not from $\GViews(M^i,P_E)$ and the threshold of ``heaviness'' is $\frac{\eps}{2n_B}$ rather than $\frac{\eps}{n_B}$.) Similarly, we call $q$ a green query of Bob's type, denoted $q \in \GRB$, if $q$ is not red and $ \Pr_{(V^i_A,V^i_B) \gets \Views(M^i,P_E)}[q \in \Q(V^i_B)] \geq \frac{\eps}{2n_A}$. We also let the set of all green queries to be $\GR = \GRA \cup \GRB$.
\end{definition}

The following claim shows that each of Eve's queries is either red or green.

\begin{claim}\label{clm:allColored}
Every query $q$ asked by Eve is either in $\RED$ or in $\GR$.
\end{claim}

\begin{proof}
  If $q$ is a query of Eve which is not red, then $\Pr_{\Views(M^i,P_E)}[\Good(M^i,P_E)] \geq 1/2$ where $(M^i,P_E)$ is the view of Eve when asking $q$. Since Eve is asking $q$, either of the following holds:
  \begin{enumerate}
    \item $\Pr_{(V^i_A,V^i_B) \gets \GViews(M^i,P_E)}[q \in \Q(V^i_A)] \geq \frac{\eps}{n_B}$, or
    \item $\Pr_{(V^i_A,V^i_B) \gets \GViews(M^i,P_E)}[q \in \Q(V^i_B)] \geq \frac{\eps}{n_A}.$
  \end{enumerate}
  If case 1 holds, then
  \begin{align*}
  \Pr_{(V^i_A,V^i_B) \gets \Views(M^i,P_E)}[q \in \Q(V^i_A)]
  &\geq \Pr_{(V^i_A,V^i_B) \gets \Views(M^i,P_E)}[\Good(M^i,P_E) \land q \in \Q(V^i_A)] \\
  &= \Pr_{\Views(M^i,P_E)}[\Good(M^i,P_E)] \cdot  \Pr_{(V^i_A,V^i_B) \gets \GViews(M^i,P_E)}[q \in \Q(V^i_A)] \\
  &\geq (\frac{1}{2}) \cdot
   \frac{\eps}{n_B} = \frac{\eps}{2 n_B}
  \end{align*}
  which implies that $q \in \GRA$.  Case 2  similarly shows that $q \in \GRB$.
\end{proof}

We will bound the size of the queries of each color separately.

\begin{claim}[Bounding Red Queries] \label{clm:notRed}
$\Pr_\EXEC[\RED \ne \es] \leq 6\eps$.
\end{claim}

\begin{claim}[Bounding Green Queries] \label{clm:notGreen}
$\Ex_\EXEC[|\GR|] \leq 4 n_A \cdot n_B / \eps$. Therefore, by Markov inequality, $\Pr_\EXEC[|\GR| \geq n_A \cdot n_B / \eps^2] \leq 4\eps$.
\end{claim}

\paragraph{Proving Lemma~\ref{lem:EveEff}.} Lemma~\ref{lem:EveEff} follows by a  union bound and Claims~\ref{clm:allColored},~\ref{clm:notRed}, and~\ref{clm:notGreen}.

\begin{proof}[Proof of Claim~\ref{clm:notRed}]
Claim~\ref{clm:notRed} follows directly from Lemma~\ref{lem:IR} and Lemma~\ref{lem:success} as follows.
Let $\bfx$ (in Lemma~\ref{lem:IR}) be $\EXEC$, the event $E$ be $\Fail$, the sequence $\bfx_1,\dots,$ be the sequence of pieces of information that Eve receives (\ie the messages and oracle answers), $\lambda = 3\eps$, $\lambda_1 = 1/2$ and $\lambda_2 = 6\eps$. Lemma~\ref{lem:success} shows that $\Pr[\Fail] \leq \lambda$. Therefore, if we let $D$ be the event that at some point conditioned on Eve's view the probability of $\Fail$ is more than $\lambda_1$, Lemma~\ref{lem:IR} shows that the probability of $D$ is at most $\lambda_2$. Also note that for every sampled $(M,P_E)$, $\Pr[\neg \Good \mid (M,P_E)] \leq \Pr[\Fail \mid (M,P_E)]$. Therefore, with probability at least $1-\lambda_2 = 1-6\eps$, during the execution of the system, the probability of $\Good(M,P_E)$ conditioned on Eve's view will never go below $1/2$.
\end{proof}

\begin{proof}[Proof of Claim~\ref{clm:notGreen}]

We will prove that $\Ex_\EXEC[|\GRA|] \leq 2 n_A \cdot n_B / \eps$, and $\Ex_\EXEC[|\GRB|] \leq 2 n_A \cdot n_B / \eps$ follows symmetrically. Using these two upper bounds we can derive Claim~\ref{clm:notGreen} easily.

For a fixed query $q \in \bits^{\ell}$, let $I_q$ be the event, defined over $\EXEC$, that Eve asks $q$ as a green query of Alice's type (\ie $q \in \GRA$). Let $F_q$ be the event that Alice actually asks $q$ (\ie $q \in Q_A$). By linearity of expectation we have  $\Ex_\EXEC[|\GRA|] = \sum_q \Pr[I_q]$ and $\sum_q \Pr[F_q]  \leq |Q_A| \leq n_A$. Let $\gamma = \frac{\eps}{2 n_B}$.
We claim that for all $q$ it holds that:
\begin{equation}\label{eq:7}
\Pr[I_q] \cdot \gamma \leq \Pr[F_q].
\end{equation}
First note that Inequality (\ref{eq:7}) implies Claim~\ref{clm:notGreen} as follows:
$$\Ex_\EXEC[|\GRA|]  = \sum_q \Pr[I_q] \leq \frac{1}{\gamma}\sum_q \Pr[F_q]  \leq \frac{n_A}{\gamma} = \frac{2n_A n_B}{\eps}.$$

To prove Inequality (\ref{eq:7}), we use Lemma~\ref{lem:IR} as follows. The underlying random variable $\bfx$ (of Lemma~\ref{lem:IR}) will be $\EXEC$, the event $E$ will be $F_q$,  the sequence of random variables $\bfx_1,\bfx,\dots$  will be the sequence of pieces of information that Eve observes, $\lambda$ will be $\Pr[F_q]$, and $\lambda_1$ will be $\gamma$. If $I_q$ holds, it means that based on Eve's view the query $q$ has at least $\gamma$ probability of being asked by Alice (at some point before), which implies that the event $D$ (of Lemma~\ref{lem:IR}) holds, and so $I_q \se D$. Therefore, by Lemma~\ref{lem:IR}  $\Pr[I_q] \leq \Pr[D] \leq \lambda/\lambda_1 = \Pr[F_q]/\gamma$ proving Inequality (\ref{eq:7}).
\end{proof}

\begin{remark}[Sufficient Condition for Efficiency of Eve] \label{rem:Eff}
The proof of Claims~\ref{clm:allColored} and~\ref{clm:notGreen} only depend on the fact that all the queries asked by Eve are are either $(\eps/n_B)$-heavy for Alice or $(\eps/n_A)$-heavy for Bob with respect to the distribution $\GViews(M,P)$. More formally, all we need is that whenever Eve asks a query $q$ it holds that
$$\Pr_{(V_A,V_B) \gets \GViews(M,P)}[q \in \Q(V_A)] \geq \frac{\eps}{n_B} \text{~~~~or~~~~} \Pr_{(V_A,V_B) \gets \GViews(M,P_)}[q \in \Q(V_B)] \geq \frac{\eps}{n_A}.$$
In particular, the conclusions of Claims~\ref{clm:allColored} and~\ref{clm:notGreen} hold regardless of which heavy queries Eve chooses to ask at any moment, and the only important thing is that all the queries asked by Eve were heavy at the time of being asked.
\end{remark}

\section{Extensions} \label{sec:Extensions}

In this section we prove several extensions to our main result that can all be directly obtained from the results proved in Section~\ref{sec:desc}. The main goal of this section is to generalize our main result to a broader setting so that it could be applied in subsequent work more easily.  We assume the reader is familiar with the definitions given in Sections \ref{sec:prelims} and \ref{sec:desc}.

\subsection{Making the Views Almost Independent}
In this section we will prove Theorem~\ref{thm:indep} along with several other extensions.
These extensions were used in~\cite{DachmanLMM11} to prove black-box separations for certain optimally-fair coin-tossing protocols. We first mention these extensions informally and then will prove them formally.

 \begin{description}
 \item[Average Number of Queries:] We will show how to decrease the  number of queries  asked by Eve by a factor of $\Omega(\eps)$ if we settle for bounding the \emph{average} number of queries asked by Eve. This can always be turned into a an attack of worst-case complexity by putting the $\Theta(\eps)$ multiplicative factor back and applying the Markov inequality.
 \item[Changing the Heaviness Threshold:] We will show that the attacker Eve of Construction~\ref{const:Eve} is ``robust'' with respect to choosing its ``heaviness'' parameter $\eps$. Namely, if she changes the parameter $\eps$ arbitrarily during her attack, as long as $\eps \in [\eps_1,\eps_2]$ for some $\eps_1 < \eps_2$, we can still show that Eve is both ``successful'' and ``efficient'' with high probability.

 \item[Learning the Dependencies:] We will show that our adversary Eve can, with high probability, learn the ``dependency'' between the views of Alice and Bob in any two-party computation. Dachman \etal~\cite{DachmanLMM11} were the first to point out that such results can be obtained from results proved in original publication of this work~\cite{BarakM09}. Haitner \etal~\cite{HaitnerOZ12}, relying some of the results proved in~\cite{BarakM09},  proved a variant of the first part of our Theorem~\ref{thm:indep} in which $n$ bounds \emph{both} of $n_A$ and $n_B$.

 \item[Lightness of Queries:] We  observe that with high probability the following holds at the end of every round conditioned on Eve's view: For every query $q$ \emph{not} learned by Eve, the probability of $q$ being asked  by Alice or Bob remains ``small''. Note that here we are \emph{not} conditioning on the event $\Good(M,P)$.

 \end{description}

Now we formally prove the above extensions.

The following definition defines a \emph{class} of attacks that share a specific set of properties.

\begin{definition} \label{def:eps1eps2}
For $\eps_1 \leq \eps_2$, we call Eve an $(\eps_1,\eps_2)$-attacker, if Eve performs her attack in the framework of Construction~\ref{const:Eve}, but instead of using a single parameter $\eps$ it uses $\eps_1\leq \eps_2$ as follows.
\begin{enumerate}
  \item {\bf All queries asked are heavy according to parameter $\eps_1$.} Every query $q$ asked by Eve, at the time of being asked, should be either $(\eps_1/n_B)$-heavy for Alice or $(\eps_1/n_A)$-heavy for Bob with respect to the distribution $\GViews(M,P)$ where $(M,P)$ is the  view of Eve when asking $q$.
  \item {\bf No heavy query,  as parameterized by $\eps_2$, remains unlearned.}  At the end of every round $i$, if $(M,P)$ is the view of Eve at that moment, and  if $q$ is any query that is either $(\eps_2/n_B)$-heavy for Alice or $(\eps_2/n_A)$-heavy for Bob with respect to the distribution $\GViews(M,P)$, then Eave has to have learned that query already to make sure $q \in \Q(P)$.
\end{enumerate}
\end{definition}

\paragraph{Comparison with Eve of Construction \ref{const:Eve}.} The Eve of Construction \ref{const:Eve} is an $(\eps,\eps)$-attacker, but for $\eps_1<\eps_2$ the class of $(\eps_1,\eps_2)$-attackers include algorithms that could not  necessarily be described by Construction \ref{const:Eve}. For example, an $(\eps_1,\eps_2)$-attackers can chose any $\eps \in [\eps_1,\eps_2]$ and run the attacker of Construction \ref{const:Eve} using parameter $\eps$, or it can even keep changing its parameter $\eps \in [\eps_1,\eps_2]$ \emph{along the execution} of the attack. In addition, the attacker of Construction \ref{const:Eve} needs to choose the \emph{lexicographically first} heavy query, while an $(\eps_1,\eps_2)$-attacker has the freedom of choosing \emph{any} query so long as it is $(\eps_1/n_B)$-heavy for Alice or $(\eps_1/n_A)$-heavy for Bob. Finally, an $(\eps_1,\eps_2)$-attacker could use its own randomness $r_E$ that affects its choice of queries, as long as it respects the two conditions of Definition \ref{def:eps1eps2}.

\begin{definition}[Self Dependency]
 For every joint distribution $(\bfx,\bfy)$, we call $\SelDep(\bfx,\bfy)= \Delta((\bfx,\bfy), (\bfx \times \bfy))$ the \emph{self (statistical) dependency} of a  $(\bfx,\bfy)$
 where in $(\bfx \times \bfy)$ we sample $\bfx$ and $\bfy$ independently from their marginal distributions.
\end{definition}

The following theorem formalizes Theorem~\ref{thm:indep}. The last part of the theorem is used by~\cite{DachmanLMM11} to prove lower-bounds on coin tossing protocols from one-way functions. We advise the reader to review the notations of Section \ref{sec:Notation} as we will use some of them here for our modified variant of $(\eps_1,\eps_2)$-attackers.

\begin{theorem}[Extensions to Main Theorem] \label{thm:extensions}
Let, $\Pi, r_A, n_A, r_B, n_B, H, s_A, s_B, \rho$ be as in Theorem~\ref{thm:mainFormal} and suppose $\eps_1 \leq \eps_2 < 1/10$. Let Eve be \emph{any} $(\eps_1,\eps_2)$-attacker who is modified to stop asking any queries as soon as she is about to ask a red query (as defined in Definition \ref{def:color}). Then the following claims hold.
\begin{enumerate}
  \item {\bf Finding outputs:} Eve's output agrees with Bob's output with probability $\rho - 16 \eps_2$.

  \item {\bf Average number of queries:} The expected number of queries asked by Eve is at most $4 n_A n_B / \eps_1$. More generally, if we let $Q_\eps$ to be the number of (green) queries that are asked because of being $\eps$-heavy for a fixed $\eps \in [\eps_1,\eps_2]$, it holds that $\Ex[|Q_\eps|] \leq 4 n_A n_B / \eps$.

  \item {\bf Self-dependency at every fixed round.} For any fixed round $i$, it holds that
      $$\Ex_{(M,P) \gets (\bfM^i,\bfP^i_E)}[\SelDep(\Views(M,P))]  \leq 21 \cdot \eps_2.$$

  \item {\bf Simultaneous self-dependencies at all rounds.} For every $\alpha,\beta$ such that $0<\alpha<1$, $0<\beta<1$, and $\alpha \cdot \beta \geq \eps_2$, with probability at least $1-9\alpha$ the following holds: at the end of \emph{every} round $i$, we have $\SelDep(\Views(M^i,P^i_E)) \leq 9\beta$.

  \item {\bf Simultaneous lightness at all round.} For every $\alpha,\beta$ such that $0<\alpha<1$, $0<\beta<1$, and $\alpha \cdot \beta \geq \eps_2$, with probability at least $1-9\alpha$ the following holds: at the end of \emph{every} round, if $q \nin \Q(P)$ is any query not learned by Eve so far we have
$$\Pr_{(V_A,V_B) \gets \Views(M,P)}[q \in \Q(V_A)] < \frac{\eps_2}{n_B} + \beta \text{~~~~and~~~~} \Pr_{(V_A,V_B) \gets \Views(M,P)}[q \in \Q(V_B)] < \frac{\eps_2}{n_A} + \beta.$$  

\item {\bf Dependency and lightness at every fixed round.} For every round $i$ and  every $(M,P) \gets (\bfM^i,\bfP^i_E) $ there is a product distribution $(\bfW_A \times \bfW_B)$  such that the following two hold:
\begin{enumerate}
\item $\Ex_{(M,P)} [\Delta(\Views(M,P),(\bfW_A \times \bfW_B))] \leq 15 \eps_2$.
\item With probability $1-6\eps_2$ over the choice of $(M,P)$ (which determines the distributions $\bfW_A,\bfW_B$ as well), we have $\Pr[q \in \Q(\bfW_A)] < \frac{\eps_2}{n_B}$ and $\Pr[q \in \Q(\bfW_B)] < \frac{\eps_2}{n_A}$.
\end{enumerate}
\end{enumerate}
\end{theorem}

In the rest of this section we prove Theorem~\ref{thm:extensions}.
To prove all the properties, we first assume that the adversary is an $(\eps_1,\eps_2)$-attacker, denoted by UnbEve (Unbounded Eve), and then will analyze how stopping UnbEve upon reaching a red query (\ie converting it into Eve) will affect her execution.

Remarks~\ref{rem:GraphChar} and~\ref{rem:Eff} show that many of the results proved in the previous section extend to the more general setting of $(\eps_1,\eps_2)$-attackers.

\begin{claim} \label{clm:extensions}
All the following lemmas, claims, and corollaries still hold when we use an arbitrary $(\eps_1,\eps_2)$-attacker and $\eps_1<\eps_2<1/10$:
\begin{enumerate}
  \item \label{item:lem:combChar} Lemma~\ref{lem:combChar} using $\eps=\eps_2$.

  \item \label{item:cor:statClose} Corollary~\ref{cor:statClose}  using $\eps=\eps_2$.

  \item \label{item:lem:success} Lemma~\ref{lem:success}  using $\eps=\eps_2$.

  \item \label{item:lem:EveFinds} Lemma~\ref{lem:EveFinds}  using $\eps=\eps_2$.

  \item \label{item:clm:notRed} Claim~\ref{clm:notRed}  using $\eps=\eps_2$.

  \item \label{item:clm:allColored} Claim~\ref{clm:allColored} by using $\eps=\eps_1$ in the definition of green queries.

  \item \label{item:clm:notGreen} Claim~\ref{clm:notGreen} by using $\eps=\eps_1$ in the definition of green queries. More generally, the proof of Claim~\ref{clm:notGreen} works directly (without any change)  if we run a $(\eps_1,\eps_2)$ attack, but define the green queries using a parameter $\eps \in [\eps_1,\eps_2]$ (and only count such queries, as green ones).
\end{enumerate}
\end{claim}

\begin{proof}
Item~\ref{item:lem:combChar} follows from Remark~\ref{rem:GraphChar} and the the second property of $(\eps_1,\eps_2)$-attackers. All Items~\ref{item:cor:statClose}--\ref{item:clm:notRed} follow from Item~\ref{item:lem:combChar} because the proofs of the corresponding statements in previous section \emph{only} rely (directly or indirectly) on Lemma~\ref{lem:combChar}.

Items~\ref{item:clm:allColored} and~\ref{item:clm:notGreen} follow from Remark~\ref{rem:Eff} and the first property of $(\eps_1,\eps_2)$-attackers.
\end{proof}

\paragraph{Finding Outputs.} By Item~\ref{item:lem:EveFinds} of Claim~\ref{clm:extensions}, UnbEve hits Bob's output with probability at least $\rho-10\eps_2$. By Item~\ref{item:clm:notRed} of Claim~\ref{clm:extensions}, the probability that UnbEve asks any red queries is at most $6 \eps_2$.
Therefore,  Eve's output will agree with Bob's output with probability at least $\rho-10\eps-6\eps =\rho-16\eps$.

\paragraph{Number of Queries.} By Item~\ref{item:clm:notGreen}, the expected number of green queries asked by UnbEve is at most $4 n_A n_B / \eps_1$. As also specified in Item~\ref{item:clm:notGreen}, the more general upper bound, for an arbitrary parameter $\eps \in [\eps_1,\eps_2]$, holds as well.

\paragraph{Dependencies.} We will use the following definition which relaxes the notion of self dependency by computing the statistical distance of $(\bfx,\bfy)$ to the closest product distribution (that might be different from $(\bfx \times \bfy)$).

\begin{definition}[Statistical Dependency]
For two jointly distributed random variables $(\bfx,\bfy)$, let the \emph{statistical dependency} of $(\bfx,\bfy)$, denoted by $\Dep(\bfx,\bfy)$, be the minimum statistical distance of $(\bfx,\bfy)$ from all product distributions defined over $\Supp(\bfx) \times \Supp(\bfy)$. More formally:
 $$\Dep(\bfx,\bfy) = \inf_{(\bfa \times \bfb)} \Delta((\bfx,\bfy), (\bfa \times \bfb))$$
 in which $\bfa \times \bfb$ are distributed over $\Supp(\bfx) \times \Supp(\bfy)$.
\end{definition}

By definition, we have $\Dep(\bfx,\bfy) \leq \SelDep(\bfx,\bfy)$. The following lemma by~\cite{MahmoodyMP12} shows that the two quantities can not be too far.

\begin{lemma}[Lemma A.6 in~\cite{MahmoodyMP12}] \label{lem:MMP}
  $\SelDep(\bfx,\bfy) \leq 3 \cdot \Dep(\bfx,\bfy)$.
\end{lemma}

\begin{remark}
We note that, $\SelDep(\bfx,\bfy)$ can, in general, be larger than $\Dep(\bfx,\bfy)$. For instance consider the following joint distribution over $(\bfx,\bfy)$ where $\bfx$ and $\bfy$ are both Boolean variables: $\Pr[\bfx=0,\bfy=0]=1/3, \Pr[\bfx=1,\bfy=0]=1/3, \Pr[\bfx=1,\bfy=1]=1/3, \Pr[x=0,y=1]=0$.
It is easy to see that $\SelDep(\bfx,\bfy) = 2/9$, but $\Delta((\bfx,\bfy), (\bfa \times \bfb)) = 1/6 < 2/9$ for a product distribution $(\bfa \times \bfb)$ defined as follows: $\bfa \equiv \bfx$ and $\Pr[\bfb=0]=\Pr[\bfb=1]=1/2$.
\end{remark}

The following lemma follows from Lemma~\ref{lem:nestedDistance} and the definition of statistical dependency.
\begin{lemma} \label{lem:DepChain}
For jointly distributed $(\bfx,\bfy)$ and event $E$ defined over the support of $(\bfx,\bfy)$, it holds that $\Dep(\bfx,\bfy) \leq \Pr_{(\bfx,\bfy)}[E] + \Dep((\bfx,\bfy) \mid \neg E)$. We take the notational convention that whenever $ \Pr_{(\bfx,\bfy)}[E]=0$ we let $\Dep((\bfx,\bfy) \mid \neg E)=1$.
\end{lemma}

\begin{proof}
Let $(\bfa \times \bfb)$ be such that $\Delta(((\bfx,\bfy) \mid \neg E), (\bfa \times \bfb)) \leq \delta$. For the same $(\bfa \times \bfb)$, by Lemma~\ref{lem:nestedDistance} it holds that
$\Delta((\bfx,\bfy), (\bfa \times \bfb)) \leq \Pr_{(\bfx,\bfy)}[E] + \delta$. Therefore

\begin{align*}
 \Dep(\bfx,\bfy) = \inf_{(\bfa \times \bfb)} \Delta((\bfx,\bfy), (\bfa \times \bfb))
&\leq \Pr_{(\bfx,\bfy)}[E] +  \inf_{(\bfa \times \bfb)} \Delta(((\bfx,\bfy) \mid \neg E), (\bfa \times \bfb))  \\
&\leq \Pr_{(\bfx,\bfy)}[E] + \Dep((\bfx,\bfy) \mid \neg E) .
\end{align*}
\end{proof}

\paragraph{Self-dependency at every fixed round.} By Item~\ref{item:cor:statClose} of Claim~\ref{clm:extensions}, we get that by running UnbEve we obtain $\Dep(\GViews(M,P))\leq 2 \eps_2$ where $(M,P)$ is the view of UnbEve at the end of the protocol. By also Lemma~\ref{lem:DepChain} we get:
\begin{align*}
\Dep(\Views(M,P))
&\leq \Pr_\EXEC[\neg \Good \mid (M,P)] + \Dep(\GViews(M,P)) \\
&\leq \Pr_\EXEC[\neg \Good \mid (M,P)] + 2\eps_2.
\end{align*}
 Therefore, by Item~\ref{item:lem:success} of Claim~\ref{clm:extensions} and Lemma~\ref{lem:MMP} we get
\begin{align*}
\Ex_{(M,P)\gets(\bfM,\bfP)}[\Dep(\Views(M,P))]
&\leq 3 \cdot \left( \Ex_{(M,P)\gets(\bfM,\bfP)}\left[\Dep(\Views(M,P))\right] \right) \\
&\leq 3 \cdot \left(  \Ex_{(M,P)\gets(\bfM,\bfP)} \left[\Pr_\EXEC[ \neg \Good \mid (M,P)]\right] + 2\eps_2 \right)  \\
& \leq 3 \cdot \left(  \Pr_\EXEC[ \neg \Good] + 2\eps_2 \right) \leq 3 \cdot 5\eps_2 = 15 \eps_2
\end{align*}

Since the probability of UnbEve asking any red queries is at most $6\eps_2$ (Item~\ref{item:clm:notRed} of Claim~\ref{clm:extensions}), therefore when we run Eve, it holds that $\Ex_{(M,P) \gets (\bfM,\bfP)}[\Dep(\Views(M,P))] $  increases at most by  $6\eps_2$ compared to when running UnvEve. This is because  whenever we halt the execution of  Eve (which happens with probability at most $6\eps_2$)  this can lead to statistical dependency of $\Views(M,P)$ at most $1$. Therefore, if we use Eve instead of UnbEve, it  holds that
 $$\Ex_{(M,P)\gets(\bfM,\bfP)}[\Dep(\Views(M,P))]  \leq 15\eps_2 + 6\eps_2 = 21 \eps_2.$$

\paragraph{\bf Simultaneous self-dependencies at all rounds.}
First note that $0<\alpha<1$, $0<\beta<1$, and $\alpha \cdot \beta \geq \eps_2$ imply that $\alpha \geq \eps_2$ and $\beta \geq \eps_2$. By Item~\ref{item:lem:success} of Claim~\ref{clm:extensions}, when we run UnbEve, it holds that $\Pr_\EXEC[\Fail] \leq 3\eps_2$, so by Lemma~\ref{lem:IR} we conclude that with probability at least $1-3\alpha$ it holds that during the execution of the protocol, the probability of $\Fail$ (and thus, the probability of $\neg \Good(M,P)$) conditioned on Eve's view always remains at most $\beta$.  Therefore, by Item~\ref{item:cor:statClose} of Claim~\ref{clm:extensions} and Lemma~\ref{lem:DepChain}, with probability at least $1-3\alpha$ the following holds  at the end of \emph{every} round (where $(M,P)$ is Eve's view at the end of that round)
\begin{align*}
\Dep(\Views(M,P))
&\leq \Pr_\EXEC[\neg \Good \mid (M,P)] + \Dep(\GViews(M,P)) \\
&\leq \beta + 2\eps_2 \leq 3 \beta.
\end{align*}

Using Lemma~\ref{lem:MMP} we obtain the bound $\SelDep(\Views(M,P)) \leq 9 \beta$.
Since the probability of UnbEve asking any red queries is at most $6 \eps_2$,  by a union bound we conclude that with probability at least $1-3\alpha-6\eps_2 > 1-9 \alpha$,   we still get $\SelDep(\Views(M,P)) \leq 9\beta$  at the end of every round.

\paragraph{Simultaneous lightness at all rounds.}  As shown in the previous item, for such $\alpha,\beta$, with probability at least $1-9\alpha$ it holds that during the execution of the protocol, the probability of $\Fail$ (and thus, the probability of $\neg \Good(M,P)$) conditioned on Eve's view always remains at most $\beta$. Now suppose $(M,P)$ be the view of Eve at the end of some round where $\Pr_{\Views(M,P}[\neg \Good(M,P)] \leq \beta$. By the second property of $(\eps_1,\eps_2)$-attackers, it holds that:
$$\Pr_{(V_A,V_B) \gets \Views(M,P)}[q \in \Q(V_A)] \leq
\Pr_{\Views(M,P)}[\neg \Good(M,P)] +
\Pr_{(V_A,V_B) \gets \GViews(M,P)}[q \in \Q(V_A)] \leq \eps_2/n_B +\beta. $$
The same proof shows that a similar statement holds for Bob.

\paragraph{Dependency and lightness at every fixed round.} Let $(\bfW_A,\bfW_B) \equiv \GViews(M,P)$. The product distribution we are looking for will be $\bfW_A \times \bfW_B$. When we run UnbEve,  by Lemma~\ref{lem:success} it holds that $\Ex_{(M,P)}[\Delta((\bfW_A,\bfW_B),\Views(M,P))] \leq 3\eps_2$, because otherwise the probability of $\Fail$ will be more than $3\eps_2$. Also, by Corollary~\ref{cor:statClose} it holds that $\Dep(\Views(M,P)) \leq 2\eps_2$, and by Lemma~\ref{lem:MMP}, it holds that $\SelDep(\Views(M,P)) = \Delta(\Views(M,P), (\bfW_A \times \bfW_B)) \leq 6 \eps_2$. Thus, when we run UnbEve, we get $\Ex_{(M,P)}[\Delta((\bfW_A \times \bfW_B),\Views(M,P))] \leq 9 \eps_2$. By Claim~\ref{clm:notRed}, the upper bound of $9 \eps_2$ when we modify  UnbEve to Eve (by not asking red queries), could increase only by $6 \eps_2$. This proves the first part.

To prove the second part, again we use Claim~\ref{clm:notRed} which bounds the probability of asking a red query by $6 \eps_2$. Also, as long as  we do not halt  Eve (\ie no red query is asked), Eve and UnbEve remain the same, and the lightness claims hold for UnbEve by definition of the attacker UnbEve.

\subsection{Removing the Rationality Condition} \label{sec:removeRational}
In this subsection we show that \emph{all} the results of this paper, except the graph characterization of Lemma~\ref{lem:combChar}, hold even with respect to random oracles that are not necessarily rational according to Definition~\ref{def:RandOr}. We will show that a variant of Lemma~\ref{lem:combChar}, which is  sufficient for all of our applications, still holds. In the following, by an \emph{irrational random oracle} we refer to a random oracle that satisfies Definition~\ref{def:RandOr} except that its probabilities might not be rational.


\begin{lemma}[Characterization of $\Views(M,P)$] \label{lem:Char}
Let $H$ be an irrational oracle,  let $M$ be the sequence of messages sent between Alice and Bob so far, and let $P$ be the set of oracle query-answer pairs known to Eve (who uses parameter $\eps$) by the end of the round in which the last message in  $M$ is sent. Also suppose $\Pr_{\Views(M,P)}[\Good(M,P)]>0$.
Let $(\bfV_A, \bfV_B)$ be the joint view of Alice and Bob as sampled from $\GViews(M,P)$, and let $\cU_A = \Supp(\bfV_A), \cU_B = \Supp(\bfV_B)$.
Let $G = (\cU_A,\cU_B,E)$ be a bipartite graph with vertex sets $\cU_A,\cU_B$ and connect $u_A \in \cU_A$ to $u_B \in \cU_B$  if and only if $\Q(u_A) \cap \Q(u_B) \se \Q(P)$. Then there exists a distribution $\bfU_A$ over $\cU_A$ and a distribution $\bfU_B$ over $\cU_B$ such that:
\begin{enumerate}

  \item \label{item:Degrees2} For every vertex $u \in \cU_A$, it holds that $\Pr_{v \gets \bfU_B}[u \not \sim v] \leq 2\eps$, and similarly for every vertex $u \in \cU_B$, it holds that $\Pr_{v \gets \bfU_A}[u \not \sim v] \leq 2\eps$.

  \item \label{item:EquivDists2} The distribution $(V_A,V_B) \gets \GViews(M,P)$ is identical to: sampling $u \gets \bfU_A$ and $v \gets \bfU_B$   \emph{conditioned on} $u \sim v$, and outputting the views corresponding to $u$ and $v$.

\end{enumerate}
\end{lemma}

\begin{proof}[Proof Sketch]
The distributions $\bfU_A$ and $\bfU_B$ are in fact the same as the distributions $\bfA$ and $\bfB$ of Lemma~\ref{lem:product}. The rest of the proof is identical to that of Lemma~\ref{lem:combChar} \emph{without any} vertex repetition. In fact, repetition of vertices (to make the distributions uniform)  cannot be necessarily done anymore because of the irrationality of the probabilities. Here we explain the alternative parameter that takes the role of $|E^{\not\sim}(u)|/|E|$. For $u \in \cU_A$ let $q^{\not\sim}(u)$ be the probability that if we sample an edge $e \gets (\bfV_A,\bfV_B)$, it does not contain $u$ as Alice's view, and define $q^{\not\sim}(u)$ for $u \in \cU_B$ similarly. It can be verified that by the very same argument as in Lemma~\ref{lem:combChar}, it holds that $q^{\not\sim}(u) \leq \eps$ for every vertex $u$ in $G$. The other steps of the proof remain the same.
\end{proof}

The characterization of $\Views(M,P)$ by Lemma~\ref{lem:Char} can be used to derive Corollary~\ref{cor:statClose} directly (using the same distributions $\bfU_A$ and $\bfU_B$). Remark~\ref{rem:GraphChar} also holds \wrt Lemma~\ref{lem:Char}.
Here we show how to derive Lemma~\ref{lem:success2} and the rest of the results will follow immediately.

\paragraph{Proving Lemma~\ref{lem:success2}.} Again, we prove Lemma~\ref{lem:success2} even conditioned on choosing any vertex $v$ that describes Bob's view. For such vertex $v$,  the distribution of Alice's view, when we choose a random edge $(u,v') \gets (\bfV_A,\bfV_B)$ conditioned on $v=v'$ is the same as choosing $u \gets \bfU_A$ conditioned on $u \sim v$. Let's call this distribution $\bfU^v_A$. Let $S = \{u
\in \cU_A \mid q \in A_u \}$ where $q$ is the next query of Bob as specified by $v$. Let $p(S) = \sum_{u \in S} \Pr[\bfU_A = u], q(S) = \Pr_{(u,v) \gets (\bfV_A,\bfV_B)}[u \in S]$,  and let $p(E) = \Pr_{u \gets \bfU_A, v \gets \bfU_B} [u \sim v ]$. Also let $p^\sim(v) = \sum_{u \sim v} \Pr[\bfU_A = u]$. Then, we have:
\[
\Pr_{u \gets \bfU^v_A}[q \in A_u] \leq \frac{p(S)}{p^\sim(v)} \leq \frac{p(S)}{1-2\e} \leq \frac{p(S)}{(1-2\e) \cdot p(E)} \leq
\frac{q(S)}{(1-2\e)^2 \cdot p(E)} \leq \frac{\e}{(1-2\e)^2 \cdot n_B} < \frac{3\e}{2n_B}.
\]
The second and fourth inequalities are due to the degree lower bounds of Item~\ref{item:Degrees2} in Lemma~\ref{lem:Char}. The third inequality is because $p(E) < 1$. The fifth inequality is because of the definition of the attacker Eve who asks $\e/n_B$ heavy queries for Alice's view when sampled from $\GViews(M,P)$, as long as such queries exist.
The sixth inequality is because we are assuming $\e < 1/10$. \qed

\paragraph{\large Acknowledgement.}  We thank Russell Impagliazzo for very useful discussions and the anonymous reviewers for their valuable comments. 
\bibliographystyle{amsalpha}

\begin{thebibliography}{DSLMM11}

\bibitem[BBE92]{BennetBrEk92}
Charles~H. Bennett, Gilles Brassard, and Artur~K. Ekert, \emph{Quantum
  cryptography}, Scientific American \textbf{267} (1992), no.~4, 50--57.

\bibitem[BGI08]{BihamGoIs08}
Eli Biham, Yaron~J. Goren, and Yuval Ishai, \emph{Basing weak public-key
  cryptography on strong one-way functions}, TCC (Ran Canetti, ed.), Lecture
  Notes in Computer Science, vol. 4948, Springer, 2008, pp.~55--72.

\bibitem[BHK{\etalchar{+}}11]{BrassardHKKLS11}
Gilles Brassard, Peter H{\o}yer, Kassem Kalach, Marc Kaplan, Sophie Laplante,
  and Louis Salvail, \emph{Merkle puzzles in a quantum world}, CRYPTO (Phillip
  Rogaway, ed.), Lecture Notes in Computer Science, vol. 6841, Springer, 2011,
  pp.~391--410.

\bibitem[BKSY11]{BrakerskiKSY11}
Zvika Brakerski, Jonathan Katz, Gil Segev, and Arkady Yerukhimovich,
  \emph{Limits on the power of zero-knowledge proofs in cryptographic
  constructions}, TCC (Yuval Ishai, ed.), Lecture Notes in Computer Science,
  vol. 6597, Springer, 2011, pp.~559--578.

\bibitem[BMG09]{BarakM09}
Boaz Barak and Mohammad Mahmoody-Ghidary, \emph{Merkle puzzles are optimal - an
  {O} ( n $^2$)-query attack on any key exchange from a random oracle}, CRYPTO
  (Shai Halevi, ed.), Lecture Notes in Computer Science, vol. 5677, Springer,
  2009, pp.~374--390.

\bibitem[BR93]{BellareRo93}
Mihir Bellare and Phillip Rogaway, \emph{Random oracles are practical: {A}
  paradigm for designing efficient protocols}, ACM Conference on Computer and
  Communications Security, 1993, pp.~62--73.

\bibitem[BS08]{BrassardSa08}
Gilles Brassard and Louis Salvail, \emph{Quantum merkle puzzles}, International
  Conference on Quantum, Nano and Micro Technologies (ICQNM), IEEE Computer
  Society, 2008, pp.~76--79.

\bibitem[CGH04]{CanettiGoHa98}
Canetti, Goldreich, and Halevi, \emph{The random oracle methodology,
  revisited}, JACM: Journal of the ACM \textbf{51} (2004), no.~4, 557--594.

\bibitem[Cle86]{Cleve1986}
Richard Cleve, \emph{Limits on the security of coin flips when half the
  processors are faulty (extended abstract)}, Annual {ACM} Symposium on Theory
  of Computing (Berkeley, California), 28--30 May 1986, pp.~364--369.

\bibitem[DH76]{DiffieHe76}
Whitfield Diffie and Martin Hellman, \emph{New directions in cryptography},
  {IEEE} Transactions on Information Theory \textbf{IT-22} (1976), no.~6,
  644--654.

\bibitem[DSLMM11]{DachmanLMM11}
Dana Dachman-Soled, Yehuda Lindell, Mohammad Mahmoody, and Tal Malkin, \emph{On
  the black-box complexity of optimally-fair coin tossing}, TCC (Yuval Ishai,
  ed.), Lecture Notes in Computer Science, vol. 6597, Springer, 2011,
  pp.~450--467.

\bibitem[GGKT05]{GennaroGeKaTr05}
Rosario Gennaro, Yael Gertner, Jonathan Katz, and Luca Trevisan, \emph{Bounds
  on the efficiency of generic cryptographic constructions}, SIAM journal on
  Computing \textbf{35} (2005), no.~1, 217--246.

\bibitem[Gro96]{Grover96}
Lov~K. Grover, \emph{A fast quantum mechanical algorithm for database search},
  Annual {ACM} Symposium on Theory of Computing (STOC), 22--24 May 1996,
  pp.~212--219.

\bibitem[HHRS07]{HaitnerHoReSe07}
Iftach Haitner, Jonathan~J.\ Hoch, Omer Reingold, and Gil Segev, \emph{Finding
  collisions in interactive protocols -- {A} tight lower bound on the round
  complexity of statistically-hiding commitments}, Annual IEEE Symposium on
  Foundations of Computer Science (FOCS), IEEE, 2007, pp.~669--679.

\bibitem[Hol15]{ThomasNotes}
Thomas Holenstein, \emph{Complexity theory}, 2015,
  \url{http://www.complexity.ethz.ch/education/Lectures/ComplexityFS15/skript_printable.pdf}.

\bibitem[HOZ13]{HaitnerOZ12}
Iftach Haitner, Eran Omri, and Hila Zarosim, \emph{Limits on the usefulness of
  random oracles}, Theory of Cryptography, {TCC} (Amit Sahai, ed.), Lecture
  Notes in Computer Science, vol. 7785, Springer, 2013, pp.~437--456.

\bibitem[IR89]{ImpagliazzoRu89}
Russell Impagliazzo and Steven Rudich, \emph{Limits on the provable
  consequences of one-way permutations}, Annual ACM Symposium on Theory of
  Computing (STOC), 1989, Full version available from Russell Impagliazzo's
  home page \url{https://cseweb.ucsd.edu/~russell/secret.ps}, pp.~44--61.

\bibitem[KSY11]{KatzSY11}
Jonathan Katz, Dominique Schr{\"o}der, and Arkady Yerukhimovich,
  \emph{Impossibility of blind signatures from one-way permutations}, TCC
  (Yuval Ishai, ed.), Lecture Notes in Computer Science, vol. 6597, Springer,
  2011, pp.~615--629.

\bibitem[Mer74]{Merkle74}
Ralph~C. Merkle, \emph{{C.S.} 244 project proposal},
  \url{http://merkle.com/1974/}, 1974.

\bibitem[Mer78]{Merkle78}
Ralph~C. Merkle, \emph{{Secure communications over insecure channels}},
  Communications of the ACM \textbf{21} (1978), no.~4, 294--299.

\bibitem[MMP14]{MahmoodyMP12}
Mohammad Mahmoody, Hemanta~K Maji, and Manoj Prabhakaran, \emph{Limits of
  random oracles in secure computation}, Proceedings of the 5th conference on
  Innovations in theoretical computer science, ACM, 2014, pp.~23--34.

\bibitem[MMV11]{MahmoodyMV11}
Mohammad Mahmoody, Tal Moran, and Salil~P. Vadhan, \emph{Time-lock puzzles in
  the random oracle model}, CRYPTO (Phillip Rogaway, ed.), Lecture Notes in
  Computer Science, vol. 6841, Springer, 2011, pp.~39--50.

\bibitem[MP12]{MahmoodyP12}
Mohammad Mahmoody and Rafael Pass, \emph{The curious case of non-interactive
  commitments - on the power of black-box vs. non-black-box use of primitives},
  CRYPTO (Reihaneh Safavi-Naini and Ran Canetti, eds.), Lecture Notes in
  Computer Science, vol. 7417, Springer, 2012, pp.~701--718.

\bibitem[RSA78]{RivestShAd78}
Ronald~L. Rivest, Adi Shamir, and Leonard~M. Adleman, \emph{A method for
  obtaining digital signatures and public-key cryptosystems}, Communications of
  the ACM \textbf{21} (1978), no.~2, 120--126.

\bibitem[RTV04]{ReingoldTrVa04}
Omer Reingold, Luca Trevisan, and Salil~P. Vadhan, \emph{Notions of
  reducibility between cryptographic primitives}, TCC (Moni Naor, ed.), Lecture
  Notes in Computer Science, vol. 2951, Springer, 2004, pp.~1--20.

\end{thebibliography}
\newcommand{\etalchar}[1]{$^{#1}$}
\providecommand{\bysame}{\leavevmode\hbox to3em{\hrulefill}\thinspace}
\providecommand{\MR}{\relax\ifhmode\unskip\space\fi MR }
\providecommand{\MRhref}[2]{%
  \href{http://www.ams.org/mathscinet-getitem?mr=#1}{#2}
}
\providecommand{\href}[2]{#2}

\end{document}